\documentclass[journal]{IEEEtran}
\usepackage{xargs}
\usepackage{soul}
\usepackage{cancel}
\usepackage{mathtools}
\usepackage{verbatim}
\usepackage{enumitem}
\usepackage{graphics}
\usepackage{epsfig} 
\usepackage{amsmath,lipsum} 
\usepackage{amssymb}  
\usepackage{amsthm}
\usepackage{xfrac}
\usepackage[pdfborder={0 0 0}]{hyperref}
\graphicspath{{figures/}}
\newcommand {\matr}[2]{\left[\begin{array}{#1}#2\end{array}\right]}

\newcommand{\x}{{\mathbf{x}}}
\newcommand{\y}{{\mathbf{y}}}
\newcommand{\z}{{\mathbf{z}}}
\renewcommand{\u}{{\mathbf{u}}}
\newcommand{\vv}{{{v}}}
\newcommand{\w}{{\mathbf{w}}}
\renewcommand{\r}{{\mathbf{r}}}

\newcommand{\Wb}{\mathcal{{W}}}

\newcommand{\bx}{{\x}}
\newcommand{\bu}{{\u}}

\newcommand{\btau}{{\tau}}
\newcommand{\bv}{{\vv}}

\usepackage[usenames,dvipsnames]{xcolor}
\definecolor{wheat}{rgb}{0.96,0.87,0.70}
\definecolor{mario}{rgb}{0.8,0.8,1}
\definecolor{ivo}{rgb}{1,0.8,0.8}

\newcommand{\review}[1]{#1}

\newcommandx{\xb}[2][1=n,2=k]{\x_{#1|#2}}

\newcommandx{\zb}[2][1=n,2=k]{\bar\z_{#1|#2}}
\newcommandx{\ub}[2][1=n,2=k]{\u_{#1|#2}}
\newcommandx{\yb}[2][1=n,2=k]{\bar\y_{#1|#2}}
\newcommandx{\vb}[2][1=n,2=k]{\vv_{#1|#2}}
\newcommandx{\wb}[2][1=n,2=k]{\w_{#1|#2}}
\newcommandx{\rx}{\r^{\x}}
\newcommandx{\ru}{\r^{\u}}
\newcommandx{\tb}[2][1=n,2=k]{\tau_{#1|#2}}

\newcommandx{\xt}[2][1=n,2=k]{\tilde\x_{#1|#2}}
\newcommandx{\ut}[2][1=n,2=k]{\tilde\u_{#1|#2}}
\newcommandx{\vt}[2][1=n,2=k]{\tilde\vv_{#1|#2}}
\newcommandx{\tildet}[2][1=n,2=k]{\tilde\tau_{#1|#2}}

\newcommandx{\rbx}[2][1=n,2=k]{\r_{#1|#2}^{\x}}
\newcommandx{\rbu}[2][1=n,2=k]{\r_{#1|#2}^{\u}}

\newcommandx{\hb}[3][1=n,2=k,3={}]{h_{#1}^{#3}}
\newcommandx{\gb}[3][1=n,2=k,3={}]{g_{#1|#2}^{#3}}
\newcommandx{\gbar}[3][1=n,2=k,3={}]{\bar{g}_{#1|#2}^{#3}}
\newcommandx{\xT}[2][1=n,2=k]{\mathcal{X}_{#1|#2}}
\newcommandx{\sigb}[3][1=n,2=k,3={}]{\bar\sigma_{#1|#2}^{#3}}
\usepackage{tikz}
\usepackage{soul}

\newtheorem{Theorem}{Theorem}
\newtheorem{Lemma}{Lemma}

\newtheorem{Assumption}{Assumption}
\newtheorem{Definition}{Definition}
\newtheorem{Remark}{Remark}
\newtheorem{Example}{Example}

\begin{document}
	\title{Safe Trajectory Tracking in Uncertain Environments}
	
	\author{Ivo~Batkovic$^{1,2}$,~Paolo~Falcone$^{2,3}$,~Mohammad~Ali$^1$,~and~Mario~Zanon$^4$
		
		\thanks{This work was partially supported by the Wallenberg Artificial Intelligence, Autonomous Systems and Software Program (WASP) funded by Knut and Alice Wallenberg Foundation, and by the COPPLAR project (VINNOVA. V.P. Grant No. 2015-04849).
		}	
		\thanks{$^{1}$ Ivo Batkovic and Paolo Falcone are with the Mechatronics group at the Department of Electrical Engineering, Chalmers University of Technology, Gothenburg, Sweden {\tt\footnotesize \{ivo.batkovic,falcone\}@chalmers.se }}
		\thanks{$^{2}$ Ivo Batkovic, and Mohammad Ali are with the research department at Zenseact AB {\tt\footnotesize \{ivo.batkovic,mohammad.ali\}@zenuity.com}}
		\thanks{$^{3}$ Paolo Falcone is with the Dipartimento di Ingegneria ``Enzo Ferrari'' Universit\`a di Modena e Reggio Emilia, Italy {\tt\footnotesize falcone@unimore.it}}
		\thanks{$^{4}$ Mario Zanon is with the IMT School for Advanced Studies Lucca {\tt\footnotesize mario.zanon@imtlucca.it}}}

	\maketitle
	
	\begin{abstract}
		In Model Predictive Control~(MPC) formulations of trajectory tracking problems, infeasible reference trajectories and a-priori unknown constraints can lead to cumbersome designs, aggressive tracking, and loss of recursive feasibility. This is the case, for example, in trajectory tracking applications for mobile systems in the presence of constraints which are not fully known a-priori. In this paper, we propose a new framework called Model Predictive Flexible trajectory Tracking Control~(MPFTC), which relaxes the trajectory tracking requirement. Additionally, we accommodate recursive feasibility in the presence of a-priori unknown constraints, which might render the reference trajectory infeasible. In the proposed framework, constraint satisfaction is guaranteed at all times while the reference trajectory is tracked as good as constraint satisfaction allows, thus simplifying the controller design and reducing possibly aggressive tracking behavior. The proposed framework is illustrated with three numerical examples. 
	\end{abstract}
	
	
	\IEEEpeerreviewmaketitle
	
	\section{Introduction}\label{sec:intro}
	Model Predictive Control~(MPC) is an advanced control technique for linear, or nonlinear systems, that has been made successful by the possibility of introducing time-varying references with preview information as well as constraints.
	Standard MPC formulations penalize deviations from a setpoint or a (feasible) reference trajectory and stability and recursive feasibility guarantees have been derived for such settings~\cite{Mayne2000,rawlings2009model,borrelli2017predictive}. 
	However, in practice (a) not all constraints, that the real system could be subject to, are available at the design stage, and (b) a reference trajectory which satisfies all constraints might therefore not be available. This is the case, for example, in trajectory tracking applications for autonomous vehicles, see, e.g.,~\cite{ljungqvist2016path,ljungqvist2018stability,andersson2018receding,ljungqvist2019path}. For such problems, it would be convenient to use a reference trajectory that is easy to generate, and that is not necessarily feasible for all constraints, while detecting and enforcing a-priori unknown constraints online at all times. Existing literature on motion planning in dynamic environments~\cite{paden2016survey,mohanan2018survey} proposes ways of tackling (b), where graph-based approaches have been rather successful~\cite{fiorini1998motion,van2005roadmap,fulgenzi2008probabilistic}. In particular, in~\cite{andersson2018receding} the authors presented a unified framework, where the motion planning and control problems are solved in order to control and guide a quadcopter in various configurations with multiple moving obstacles. The authors show the benefit of being able to plan new trajectories in such environments, however, no formal safety (recursive feasibility) guarantees are provided.
			
	In this paper, we take a different point of view and deal with the presence of a-priori unknown constraints by constructing a new MPC framework which ensures safety, i.e., recursive feasibility of the system, while only requiring mild assumptions on the reference trajectory, which needs not be feasible with respect to all constraints. Consequently, we do provide safety guarantees without the need of the additional layer of motion planning. Nevertheless, our approach can still be combined with motion planning if this is deemed beneficial, e.g., to improve performance.
	
	In standard MPC, the presence of constraints, can cause the system to slow or stop, while the reference trajectory does not: in this case, as the system gets far from the reference, an undesirably aggressive behavior is obtained. On the other hand, while safety and stability are the most important requirements, a smooth, or at least not excessively aggressive, behavior is also desirable. To alleviate aggressive behaviors introduced by infeasible reference trajectories, Model Predictive Path Following Control (MPFC) has been proposed in~\cite{aguiar2005path,Faulwasser2009,kanjanawanishkul2009path,alessandretti2013trajectory,Faulwasser2016,faulwasser_implementation} by penalizing deviations from a reference path instead of a trajectory, where additional variables are introduced in order to control the position along the reference path.  The main difficulty in MPFC is the need to select an appropriate output function to define the path typically in a dimension lower than the system state space. 
	Although the results in this paper can be combined with MPFC, we also propose an alternative strategy which we call Model Predictive Flexible trajectory Tracking Control (MPFTC). Based on ideas similar to MPFC, we introduce new variables to artificially modify the time derivative of the reference trajectory in a time warping fashion. Moreover, we also discuss safety in a general sense, and give safety guarantees which build on the assumption that a safe set exists, where all (including the unknown) constraints are satisfied at all time: this is an assumption often made in practice  for stable systems at rest, as long as a safe configuration can be found~\cite{petti2005safe},~\cite{liu2017provably,beckert2017online}.

	 In order to illustrate in detail the theoretical developments, we consider two toy examples and compare our formulation with existing ones. We then design an MPFTC controller for a robotic arm, which has to follow a trajectory while avoiding an a-priori unknown obstacle. The main contributions of this paper can be summarized as follows: (a) the introduction of the MPFTC framework for flexible trajectory tracking with stability guarantees; and (b) the development of a safe framework for satisfaction of a-priori unknown constraints. 
	
	This paper is structured as follows. In Section \ref{sec:problem_description} we outline the flexible trajectory tracking problem and in Section \ref{sec:mpftc}  we prove stability for MPFTC. In Section \ref{sec:safe_mpftc} we introduce a framework with recursive feasibility guarantees for a-priori unknown constraints. We illustrate the theoretical developments in Section \ref{sec:simulations} with three numerical examples. Finally, we draw conclusions in Section \ref{sec:conclusions}.

	\subsection{Notation}
	We denote a discrete-time nonlinear system by
	\begin{equation}\label{eq:sys}
	\x_{k+1}=f(\x_k,\u_k),
	\end{equation}
	where $\x_k\in\mathbb{R}^{n_x}$ and $\u_k\in\mathbb{R}^{n_u}$ are the state and input vectors at time $k$, respectively. The state and inputs are subject to two categories of constraints: a-priori known constraints $h(\x,\u):\mathbb{R}^{n_x}\times\mathbb{R}^{n_u}\rightarrow\mathbb{R}^{n_h}$; and a-priori unknown constraints $g(\x,\u):\mathbb{R}^{n_x}\times\mathbb{R}^{n_u}\rightarrow\mathbb{R}^{n_g}$, i.e., the state and inputs must satisfy $h(\x,\u)\leq{}0$ and $g(\x,\u)\leq{}0$, where the inequalities are defined element-wise.
	
	We use the notation $\gb[n][k](\x,\u)$ to denote $g$ at time $n$, given the information available at time $k$. Moreover, we will denote by $g_n(\x,\u):=g_{n|\infty}(\x,\u)$ the real constraint, since in general $\gb[n][k](\x,\u)\neq g_n(\x,\u)$. Note that for a-priori known constraints $h_{n|k}(\x,\u):=h_n(\x,\u)$ holds by definition. We apply the same notation to state and inputs, e.g., $\x_{n|k}$ and $\u_{n|k}$ denote the predicted state and input at time $n$ given the current time $k$. In addition, to denote a set of integers, we use $\mathbb{I}_{a}^{b} := \{a,a+1,...,b\}$.

	\section{Problem Description}\label{sec:problem_description}
	
	Our aim is to control system~\eqref{eq:sys} such that both known constraints $h(\x,\u)\leq{}0$ and a-priori unknown constraints $g(\x,\u)\leq{}0$ are satisfied at all times. A situation in which some constraints cannot be known a priori occurs in the context of trajectory tracking for mobile robots, where $h$ includes, e.g., actuator saturations and/or imposes the avoidance of collision with fixed and known obstacles, while $g$ includes, e.g., moving obstacles present in the environment, whose motion trajectories are not known a-priori, but can be (over)estimated online based on measurements.
	
	Our first and essential objective is to guarantee safety of~\eqref{eq:sys}, which we define formally as
	
	\begin{Definition}[Safety]\label{def:safe}
		A controller is said to be safe in a given set $\mathcal{S}\subseteq\mathbb{R}^{n_x}$ if $\forall \, \x\in\mathcal{S}$ it generates control inputs $\mathbf{U}=\{\u_0,...,\u_\infty\}$ and corresponding state trajectories $\mathbf{X}=\{\x_0,\x_1,...,\x_\infty\}$ such that $h_k(\x_k,\u_k)\leq{}0$ and~$g_k(\x_k,\u_k)\leq{}0$, $\forall \, k \geq 0$.
	\end{Definition}
	\review{Our second objective is to control the system such that the state and input $\x_k, \u_k$ track a parametrized reference trajectory $\r(\tau)\review{:=}(\r^\x(\tau),\r^\u(\tau))$ as closely as safety allows.} If the reference parameter $\tau$ is selected to be time, its natural dynamics are given by
	\begin{equation}
	\label{eq:tau_natural}
	\tau_{k+1} = \tau_k + t_\mathrm{s},
	\end{equation}
	where $t_\mathrm{s}$ is the sampling time for sampled-data systems and $t_\mathrm{s}=1$ in the discrete-time framework. Given the presence of nonlinear dynamics and constraints, we frame the problem in the context of MPC.
	Note that if $\tau$ is forced to follow its natural dynamics~\eqref{eq:tau_natural}, then the reference tracking problem in the absence of a-priori unknown constraints $g$ is a standard MPC problem and, therefore, inherits all stability guarantees, but also a possibly aggressive behavior when the initial state is far from the reference.
	
	Approaches developed for mechanical systems in presence of large tracking errors, especially caused by reference setpoint changes, have been proposed in, e.g.,~\cite{Gros2017a}. In the setting we consider, however, perfect tracking will not be impeded by sudden setpoint changes, but rather by the presence of constraints, e.g., a mobile system might have to temporarily stop in order to avoid collisions with other systems or obstacles. Therefore, we investigate complementary approaches to those proposed in~\cite{Gros2017a}.

	One family of approaches for smooth reference tracking is the so-called MPFC~\cite{Faulwasser2009,Faulwasser2016}. While MPFC is a valid technique for tackling our problem we propose a new, alternative, approach: MPFTC, which 
	solves the problem of tracking an infeasible reference trajectory when the presence of constraints~$g(\x,\u)\leq0$ might force the system to temporarily deviate from the reference. 
	While the main difficulty in MPFC is to establish a suitable output $y=\phi(\x,\u)$ and the corresponding path, the main difficulty in MPFTC will be to pre-compute a parametrized feasible reference.

	\section{Model Predictive Flexible Tracking Control}\label{sec:mpftc}
	The main idea in MPFTC is to avoid aggressive behaviors by adapting the \review{dynamics of the} reference trajectory by means of a parameter $\tau$, which acts as a fictitious time, through relaxed dynamics given by
	\begin{equation}
		\label{eq:tau_controlled}
		\tau_{k+1} = \tau_k + t_\mathrm{s} + v_k,
	\end{equation}
	where $v$ is an additional auxiliary control input and $\tau$ becomes an auxiliary state. \review{Note that the system dynamics are unchanged and the fictitious time $\tau$ makes only the reference dynamics deviate from the natural ones. 
	}

	We formulate the MPFTC problem as the following MPC problem
	\begin{subequations}
		\label{eq:nmpc}
		\begin{align}
		\begin{split}
		\hspace{-0.5em}V(\x_k,\tau_k):=&\min_{\substack{\bx\\\btau},\substack{\bu\\\bv}}  \sum_{n=k}^{k+N-1}
		q_\r(\xb,\ub,\tb)+w \vb^2 \hspace{-20em}\\
		&\qquad\qquad+p_\r(\xb[k+N],\tb[k+N])\label{eq:nmpc_cost} \hspace{-20em}
		\end{split}\\
		\text{s.t.}\ &\xb[k][k] = \x_{k},\  \tb[k] = \tau_{k}\label{eq:nmpcState}, \\
		&\xb[n+1] = f(\xb,\ub),\label{eq:nmpcDynamics} & n\in \mathbb{I}_k^{k+N-1},\\
		&\tb[n+1] = \tb+t_\mathrm{s}+\vb, \hspace{-0.5em}& n\in \mathbb{I}_k^{k+N-1},\label{eq:tauDynamics}\\
		&\hb(\xb,\ub) \leq{} 0, \label{eq:nmpcInequality_known}& n\in \mathbb{I}_k^{k+N-1},\\
		&\gb[n][k](\xb,\ub) \leq{} 0, 
		\label{eq:nmpcInequality_unknown}& n\in \mathbb{I}_k^{k+N-1},\\
		&\xb[k+N] \in\mathcal{X}^\mathrm{f}_\r(\tb[k+N])\label{eq:nmpcTerminal},
		\end{align}
	\end{subequations}
	where $k$ is the current  time, and $N$ is the prediction horizon. In tracking MPC, typical choices for the stage and terminal costs are
	\begin{gather}
	\Delta\xb := \xb-\rx(\tb),\quad \Delta\ub:=\ub-\ru(\tb), \nonumber\\
	q_\r(\xb,\ub,\tb) = \matr{c}{\Delta\xb\\\Delta\ub}^\top{}W\matr{c}{\Delta\xb\\\Delta\ub},\label{eq:stage_cost}\\
	p_\r(\xb[k+N],\tb[k+N]) = \Delta\xb[k+N]^\top{}P\Delta\xb[k+N],\label{eq:terminal_cost}
	\end{gather}
	where $\r(\tb)=(\rx(\tb),\ru(\tb))$ is a user-provided reference trajectory. \review{Note that the cost functions $q_\r$ and $p_\r$ depend on $\tb$ only through the reference trajectory}. The matrices $W\in\mathbb{R}^{(n_x+n_u) \times (n_x+n_u)}$ and $P\in\mathbb{R}^{n_x\times n_x}$ are symmetric positive-definite matrices. Note that we use convex quadratic forms for the cost for simplicity, but the proposed framework can accommodate more general cost definitions. 
	The predicted state and controls are defined as $\xb$, $\tb$, and $\ub$, $\vb$ respectively. Constraint \eqref{eq:nmpcState} enforces that the prediction starts at the current states, and constraints \eqref{eq:nmpcDynamics}-\eqref{eq:tauDynamics}  enforce that the predicted trajectories satisfy the system dynamics. Constraints \eqref{eq:nmpcInequality_known} denote known constraints such as, e.g., actuator saturations and reference trajectory bounds, while constraint \eqref{eq:nmpcInequality_unknown} enforces constraints which are not known a-priori 
	like, e.g., constraints imposed to avoid the collision with obstacles detected by a perception layer.
	Finally, constraint \eqref{eq:nmpcTerminal} is  a terminal set, where, differently from standard formulations, the terminal constraint depends on the auxiliary state $\tb[k+N]$ relative to the reference parameter. \review{Note that, while the introduction of one additional state and control results in an increased computational complexity, such increase is typically small, since these variables have decoupled linear dynamics.}
	
	\begin{Remark}
		If the constraint $\vb=0$ is added and constraints~\eqref{eq:nmpcInequality_unknown} are not present, a standard MPC formulation is obtained. The terminal set $\mathcal{X}^\mathrm{f}_\r$ can therefore be designed as in standard MPC, where one assumes that the reference trajectory evolves according to its natural dynamics~\eqref{eq:tau_natural}. The challenges introduced by constraints~\eqref{eq:nmpcInequality_unknown} will be further analyzed in the remainder of the paper.
	\end{Remark}

	In the following, we will first prove stability for MPFTC under the standard assumptions used to prove stability for MPC, i.e., we will assume that~\eqref{eq:nmpcInequality_unknown} is inactive at the reference and does not jeopardize recursive feasibility. 
	Because constraints~\eqref{eq:nmpcInequality_unknown} cannot be known a-priori, these assumptions become unrealistic. However, this issue is typically neglected in the MPC literature, such that guaranteeing recursive feasibility and closed-loop system stability is still an open problem. In Section~\ref{sec:safe_mpftc} we will introduce a framework based on ideas typically used in robust MPC which guarantees recursive feasibility by a slight modification of the standard MPC controller design.

	In order to prove stability, we introduce the following assumptions \review{which coincide with those commonly used in MPC}, see, e.g.,~\cite{rawlings2009model,Grune2011}.
	\begin{Assumption}[System and cost regularity]\label{a:cont}
			The system model $f$ is continuous, and the stage cost \review{$q_\r:\mathbb{R}^{n_x}\times \mathbb{R}^{n_u} \times \mathbb{R}\rightarrow\mathbb{R}_{\geq{}0}$, and terminal cost $p_\r:\mathbb{R}^{n_x}\times \mathbb{R}\rightarrow\mathbb{R}_{\geq{}0}$}, are continuous at the origin and satisfy $q_\r(\rx(\tau),\ru(\tau),\tau)=0$, and $p_\r(\rx(\tau),\tau)=0$. Additionally, $q_\r(\bx_k,\bu_k,\btau_k)\geq{}\alpha_1(\|\bx_k-\rx(\btau_k)\|)$ for all feasible $\x_k$, $\u_k$, and  $p_\r(\bx_N,\btau_N)\leq\alpha_2(\|\bx_N-\rx(\btau_N)\|)$, where $\alpha_1$ and $\alpha_2$ are $\mathcal{K}_\infty$-functions.
	\end{Assumption}
	
	\begin{Assumption}[Reference feasibility] \label{a:rec_ref}
			The reference is feasible for the system  dynamics, i.e., $\r^\x(t+t_\mathrm{s})=f(\r^\x(t),\r^\u(t))$, and: 
			\begin{enumerate}[label={\textbf{\alph*)}}, ref={\ref{a:rec_ref}\alph*}]
				\item \label{a:rec_ref-A}
				the reference satisfies the known constraints \eqref{eq:nmpcInequality_known}, i.e., $\hb(\r^\x(t_n),\r^\u(t_n)) \leq{} 0$, for all $n\in\mathbb{I}_0^\infty$;
				\item \label{a:rec_ref-B}
				the reference satisfies the unknown constraints \eqref{eq:nmpcInequality_unknown}, i.e., $\gb(\r^\x(t_n),\r^\u(t_n))\leq 0$, 
				for all $n,k\in\mathbb{I}_0^\infty$.
		\end{enumerate}
	\end{Assumption}
	Assumption~\ref{a:rec_ref-B} is a strong assumption since it assumes that the reference is feasible for constraints for all future times, i.e., at time $k$ the constraint $g_{n|k+1}$ is also assumed to be satisfied. Therefore, Assumption~\ref{a:rec_ref-A} serves as a relaxed version which is more realistic and will be used later on to replace Assumption~\ref{a:rec_ref-B}.
	\begin{Assumption}[Stabilizing Terminal Conditions] \label{a:terminal}
		There exists a parametric stabilizing terminal set  $\mathcal{X}^\mathrm{f}_\r(t)$ and a terminal control law $\kappa^\mathrm{f}_\r(\mathbf{x},t)$ yielding:
		\begin{align*}
		\mathbf{x}_+^\kappa=f(\mathbf{x},\kappa^\mathrm{f}_\r(\mathbf{x},t)), && t_+ = t + t_\mathrm{s},
		\end{align*}
		such that $p_\r(\mathbf{x}_+^\kappa,t_+) - p_\r(\mathbf{x},t) \leq{} - q_\r(\mathbf{x},\kappa^\mathrm{f}_\r(\mathbf{x},t),t)$, and
		\begin{enumerate}[label={\textbf{\alph*)}}, ref={\ref{a:terminal}\alph*}]
		\item \label{a:terminal-A}$\mathbf{x}\in\mathcal{X}^\mathrm{f}_\r(t)\Rightarrow \mathbf{x}^\kappa_+\in\mathcal{X}^\mathrm{f}_\r(t_+),$ and $\hb(\mathbf{x},\kappa^\mathrm{f}_\r(\mathbf{x},t)) \leq{} 0,$ for all $n,k\in\mathbb{I}_0^\infty$;
		\item \label{a:terminal-B}
		$\x\in\mathcal{X}_\r^\mathrm{f}(t)\Rightarrow g_{n|k}(\x,\kappa_\r^\mathrm{f}(\x,t))\leq{}0$, for all $n,k\in\mathbb{I}_0^\infty$.
	\end{enumerate}
	\end{Assumption}
	Similarly to Assumption~\ref{a:rec_ref-B}, Assumption~\ref{a:terminal-B} is also difficult to verify due to the unknown constraints. Hence, the relaxed version Assumption~\ref{a:terminal-A}, which is standard in MPC settings, will be used later on to replace Assumption~\ref{a:terminal-B}. Finally, we introduce the following assumption, imposing some structure on $g$ that is needed in order to ensure that the feasibility of a solution does not become jeopardized between consecutive time instances.
	
	\begin{Assumption}[Unknown constraint dynamics] \label{a:unknown_constraints}		
		The a-priori unknown constraint functions satisfy $g_{n|k+1}(\xb,\ub) \leq g_{n|k}(\xb,\ub)$, for all $n\geq k$.
	\end{Assumption}
	\review{This assumption essentially requires the availability of a consistent characterization of the a priori unknown constraints, as we will further detail in Section~\ref{sec:coll_avoidance}.}
	
	We are now ready to prove asymptotic stability for the MPFTC framework.
	\begin{Theorem} [Nominal Asymptotic Stability]\label{prop:stab_feas}
		Suppose that Assumptions \ref{a:cont}, \ref{a:rec_ref}, \ref{a:terminal}, and \ref{a:unknown_constraints} hold, 
		and that the initial state $(\x_k,\tau_k)$ at time $k$ belongs to the feasible set of Problem \eqref{eq:nmpc}. Then the system \eqref{eq:sys}-\eqref{eq:tau_controlled} in closed loop with the solution of~\eqref{eq:nmpc} applied in receding horizon is an asymptotically stable system. \label{prop:stable}
		\begin{proof}
		    The first part of the proof follows standard arguments used to prove stability for MPC.
			By assumption, there exists an optimal control input sequence $\mathbf{U}^\star_k=\{\ub[k]^\star,...,\ub[k+N-1]^\star\}$, $\mathbf{V}^\star_k=\{\vb[k]^\star,...,\vb[k+N-1]^\star\}$ and corresponding state trajectory $\mathbf{X}_k^\star = \{\xb[k]^\star,...,\xb[k+N]^\star\}$, $\mathbf{T}_k=\{\tb[k]^\star,...,\tb[k+N]^\star\}$ at the initial time that gives the optimal value function 
			\begin{equation*}
			V(\x_k,\tau_k)=\hspace{-1.1em}\sum_{n=k}^{k+N-1}\hspace{-.9em}q_\r(\xb^\star,\hspace{-0.1em}\ub^\star,\hspace{-0.1em}\tb^\star) + w\vb^{\star{}2}+p_\r(\xb[k+N]^\star,\hspace{-0.1em}\tb^\star).
			\end{equation*}
			Applying the first control inputs $\ub[k]^\star$ and $\vb[k]^\star$, the system and auxiliary states evolve to $\x_{k+1} = f(\x_k,\ub[k]^\star)$ and $\tau_{k+1}=\tau_k+t_\mathrm{s}+\vb[k]^\star$, respectively. The sub-optimal sequences $\mathbf{U}_{k+1}=\{\ub[k+1][k]^\star,\ub[k+2]^\star,...,\kappa^\mathrm{f}_\r(\xb[k+N]^\star,\tb[k+N]^\star)\}$, and $\mathbf{V}_{k+1}=\{\vb[k+1]^\star,\vb[k+2]^\star,...,\vb[k+N]\}$, where $\xb[k+N+1]^\kappa=f(\xb[k+N]^\star,\kappa_\r^\mathrm{f}(\xb[k+N]^\star,\tb[k+N]^\star))$, $\tb[k+N+1]^\kappa=\tb[k+N]^\star+t_\mathrm{s}+\vb[k+N]$, and $\vb[k+N]=0$, are still feasible under Assumptions~\ref{a:terminal} and \ref{a:unknown_constraints}, and yield the associated cost
			\begin{align*}
			\begin{aligned}\begin{split}
			\tilde{V}(\x_{k+1},\tau_{k+1})&=\sum_{n=k+1}^{k+N} q_\r(\xb^\star,\ub^\star,\tb^\star)  + w\vb[n]^{\star{}2}\\&\qquad\qquad+ p_\r(\xb[k+N+1]^\kappa,\tb[k+N+1])
			\end{split}\\
			\end{aligned}\\
			\begin{aligned}
			&=V(\x_k,\tau_k) - q_\r(\xb[k]^\star,\ub[k]^\star,\tb[k]^\star) - w\vb[k]^{\star{}2}\\
			&\quad+ p_\r(\xb[k+N+1]^\kappa,\tb[k+N+1]^\kappa) - p_\r(\xb[k+N]^\star,\tb[k+N]^\star)\\
			&\quad+ q_\r(\xb[k+N]^\star,\kappa^\mathrm{f}_\r(\xb[k+N]^\star,\tb[k+N]^\star),\tb[k+N]^\star).
			\end{aligned}
			\end{align*}
			Using Assumption \ref{a:terminal} and optimality, the optimal value function is shown to decrease between consecutive time instances
			\begin{equation}
			V(\x_{k+1},\tau_{k+1}) 
			\leq{} V(\x_k,\tau_k) - q_\r(\xb[k]^\star,\ub[k]^\star,\tb[k]^\star).
			\end{equation}
			Assumption \ref{a:cont} entails the lower bound 
			\begin{equation}
			V(\x_k,\tau_k)\hspace{-.1em}\geq{}\hspace{-.1em}q_\r(\xb[k]^\star,\ub[k]^\star,\tb[k]^\star)
			\hspace{-.1em}\geq{}\hspace{-.1em}\alpha_1(\|\xb[k]^\star\hspace{-0.3em}-\rx(\tb[k]^\star)\|),
			\end{equation}
			and can be used to prove an upper bound \cite[Proposition 2.17]{rawlings2009model}
			\begin{equation}
			V(\x_k,\tau_k) \leq{} p_\r(\xb[k]^\star,\tb[k]^\star) \leq{} \alpha_2(\|\xb[k]^\star-\rx(\tb[k]^\star)\|).
			\end{equation}
			Therefore, the value function is a Lyapunov function and closed-loop stability follows. 
			
			So far, we have proven that the states track a reference $\r^\x(\tau_k)$ \emph{for some} $\tau_k$. Hence, it has to be shown that~$\vb[k]=0$ asymptotically as well, which is the second and non-standard part of the proof.
			We observe that (a) \review{$\rx(\tau_{k+1}) = f(\rx(\tau_k),\ru(\tau_k))$ for $\tau_{k+1}=\tau_k+t_\mathrm{s}$ holds from Assumption~\ref{a:rec_ref}}, and (b) $q_\r(\bx_k,\bu_k,\btau_k)=0$ and $p_\r(\bx_{k+N},\btau_{k+N})=0 \Rightarrow v_k=0$ by optimality. Consequently, $\lim_{k\to \infty}v_k=0$ and $\lim_{k\to \infty}\tau_{k+1}-\tau_k=t_\mathrm{s}$.
			
			The implication (b) can be proven by noting that, in case $\xb[n]^\star=\rx(\tb[n]^\star)$, $\ub[n]^\star=\ru(\tb[n]^\star)$, then $V(\x_k,\tau_k)=\sum_{n=k}^{k+N-1} w\vb[n]^{\star{}2}$ and, consequently,
			\begin{align*}
			V(\x_{k+1},\tau_{k+1}) &\leq \tilde V(\x_{k+1},\tau_{k+1}) =\sum_{n=k+1}^{k+N} w\vb[n]^{\star{}2}  \\ &\leq{} V(\x_k,\tau_k) - w\vb[k]^{\star{}2} \leq V(\x_k,\tau_k),
			\end{align*}
			where we used $\vb[k+N]=0$. 
		\end{proof}
	\end{Theorem}
	\begin{Remark} 
		Assumption~\ref{a:rec_ref} can be relaxed to only require the reference to be feasible for an unspecified $\tau_0=t_0$. In this case, the system will be stabilized to the reference with a time shift which is an integer multiple of the sampling time $t_\mathrm{s}$. If, instead, feasibility holds for all initial times, then the time shift can be any real number. As opposed to MPFTC, in standard MPC the time shift is $0$ by construction.
	\end{Remark}
	\review{\begin{Remark}\label{remark:free_time}
			Note that the initial constraint $\tb[k][k]=\tau_k$ is not necessary, and Theorem~\ref{prop:stab_feas} holds also in case the initial auxiliary state $\tb[k][k]$ is free to be selected by the optimizer. 
	\end{Remark}}
	Theorem~\ref{prop:stab_feas} proves that the proposed MPFTC formulation asymptotically stabilizes towards a reference $\r(\tau)$, while relying on assumptions which are standard in the MPC literature. 
	However, Assumptions~\ref{a:rec_ref-B},~\ref{a:terminal-B}, and~\ref{a:unknown_constraints} are difficult to enforce in practice, since they require feasibility with respect to constraints which are unknown. The next section therefore investigates how this difficulty can be tackled by replacing Assumptions~\ref{a:rec_ref-B}, and \ref{a:terminal-B} with more realistic ones, and how Assumption~\ref{a:unknown_constraints} can be verified.

	\section{Safety-Enforcing MPC}\label{sec:safe_mpftc}
	The aim of this section is to tackle the issues posed by the presence of the a-priori unknown constraints~\eqref{eq:nmpcInequality_unknown}. While we cast the problem in the framework of MPFTC, we stress that the developments proposed to enforce safety are independent of the specific tracking scenario, i.e., flexible trajectory, path, setpoint, etc., and can also be deployed in the context of MPFC proposed in~\cite{Faulwasser2009,Faulwasser2016}.
	
	\review{We propose an approach that uses ideas typically found in robust MPC. In standard robust MPC settings the uncertainty typically acts on the system that is being controlled and one can therefore use feedback to contain the uncertainty in a bounded set. However, in our setting the uncertainty is external to the system and not controllable. Since the uncertainty can grow unbounded, we will complement the usual worst-case approach with a suitably-defined safe set.}
	
		In Section~\ref{sec:coll_avoidance} we will discuss how the unknown constraints $g$ can be constructed to ensure that Assumption~\ref{a:unknown_constraints} holds. 
		Then, in Section~\ref{sec:terminal_safe_set} we will introduce relaxed terminal conditions such that also Assumption~\ref{a:terminal-B} can be dropped. Finally, in Section~\ref{sec:theorem}, we conclude the main results of the paper with a theorem. 
	
	\subsection{Predictive Collision Avoidance}
	\label{sec:coll_avoidance}
	
	In this section we discuss how the a-priori unknown constraints $\gb(\xb,\ub)$ can be constructed by relying on ideas typically used in robust MPC. Since these constraints are not known a priori, they are intrinsically related to stochastic processes. In order to guarantee constraint satisfaction at all times, one needs to assume that the stochasticity support is bounded, such that the problem of guaranteeing safety can be cast in a worst-case scenario planning. This is typically done in robust MPC, where, however, the uncertainty is present in the system and not in the constraints. We will discuss this aspect further in the remainder of this section.
	
	In order to model $\gb$ as a constraint on the worst-case scenario, we introduce function $\gamma(\x,\u,\w):\mathbb{R}^{n_x}\times\mathbb{R}^{n_u}\times\mathbb{R}^{n_w}\rightarrow\mathbb{R}^{n_g}$ and stochastic variable $\wb\in\Wb_{n|k}\subseteq\mathbb{R}^{n_w}$ with bounded support $\Wb_{n|k}$, summarizing all uncertainty related to the a-priori unknown constraints. Then we define
	\begin{equation}
	\label{eq:constr_prediction}
	\gb[n][k](\xb,\ub) := \max_{\wb\in\Wb_{n|k}} \ \gamma_{n|k}(\xb,\ub,\wb).
	\end{equation}
	This formulation implies robust constraint satisfaction, i.e.,
	\begin{align*}
	\gb[n][k](\xb,\ub) \leq 0 && \Leftrightarrow && \left \{\begin{array}{l}
	\gamma_{n|k}(\xb,\ub,\wb) \leq 0, \\ \forall \ \wb\in\Wb_{n|k}.
	\end{array}\right.
	\end{align*}	
	In a general setting, $\wb$ is the state of the dynamical system
	\begin{equation}
	\wb[n+1] = \omega(\wb,\xi_{n|k},\xb,\ub), \label{eq:constr_prediction_dynamics}
	\end{equation}
	with associated control variable $\xi_{n|k}\in\Xi\subseteq\mathbb{R}^{m_\xi}$, acting as a source of (bounded) noise. The function $\omega$ describes the dynamics, and the explicit dependence on $\xb$, $\ub$ models possible interactions between the uncertainty and system \eqref{eq:sys}. 
	\begin{Example}
		\label{ex:noise}
		In the simplest case, $\wb$ can model the sensor noise. Then,~\eqref{eq:constr_prediction_dynamics} reads as $\wb[n+1] = \xi_n$ such that $\Wb_{n|k} \equiv \Xi$ and~\eqref{eq:constr_prediction} reads as 
		\begin{align*}
		\gb(\xb,\ub) &=\max_{\wb\in\Wb_{n|k}} \ \gamma_{n|k}(\xb,\ub) + \wb, 
		\end{align*}
		i.e., the constraint has additive process noise and no dynamics are involved. 
		
		The case of process noise can be formulated as
		\begin{align*}
		\gb(\xb,\ub) &=\max_{\wb\in\Wb_{n|k}} \ \gamma_{n|k}(\xb,\ub) + \wb, \\ 
		\wb[n+1] &= \omega(\wb,\xi_{n|k}),
		\end{align*}
		where there is no interaction between the uncertainty and the controlled system. This is the case in many robust MPC formulations, see, e.g.,~\cite{Mayne2014} and references therein.	
	\end{Example}
	\begin{Remark}
		The possibility of interaction between the system and the uncertainty dynamics is introduced in order to cover multi-agent settings in which the behavior of each agent can influence the behavior of other agents, e.g., a pedestrian changing his/her trajectory because of a vehicle not yielding. 
	\end{Remark}
	
	Having introduced a model of the uncertainty dynamics, it becomes natural to rely on reachability analysis in order to predict the future evolution of  the uncertainty sets, which are then defined as outer-approximations
	\begin{align}\label{eq:reachable}
	\Wb_{n+1|k}(\xb,\ub) :\supseteq \{\, &\omega(\wb,\xi_{n},\xb,\ub)\, | \\&\hspace{3.5em} \wb\in\Wb_{n|k},\ \forall\,\xi_{n}\in\Xi \, \},\nonumber
	\end{align}
	for some initial $\mathcal{W}_{k|k}=\w_{k|k}$.
	
	In order to provide further explanation about the nature of the uncertainty sets \eqref{eq:reachable}, 
	we provide the following example.
	\begin{Example}
		\label{ex:set_consistency}
		Consider a non-cooperative, non-connected, multi-agent setting in which the behavior of the other agents is uncertain. In this case, we distinguish two types of agents: (a) the ones which are detectable by the sensors, and (b) those that are either beyond sensor range or hidden by other obstacles. For type (a), we require the model to be not underestimating the set of future states that can be reached by the other agents. For type (b), the uncertainty model must predict the possibility that an agent could appear at any moment either at the boundary of the sensor range or from behind an obstacle.
	\end{Example}
	We can now state the following result.
	\begin{Lemma}\label{lem:uncertainty}
		Suppose that $\gb$ is defined according to~\eqref{eq:constr_prediction} with $\mathcal{W}_{n|k}$ satisfying~\eqref{eq:reachable}. Then, Assumption~\ref{a:unknown_constraints} holds.
	\end{Lemma}
	\begin{proof}
		The properties of reachable sets imply $\Wb_{n|k+1}\subseteq\Wb_{n|k}$, $n\geq k+1$. Therefore
			\begin{align*}
			&\max_{\wb[n][k+1]\in\Wb_{n|k+1}} \ \gamma_{n|k}(\xb,\ub,\wb[n][k+1]) \leq{}\\ 
			&\hspace{10em}\max_{\wb\in\Wb_{n|k}} \ \gamma_{n|k}(\xb,\ub,\wb),
			\end{align*}
			since the two optimization problems have the same cost function, and the domain of the first one is not larger than the domain of the second one. Then,
			\begin{align*}
			\gb[n][k+1](\xb,\ub) \leq \gb(\xb,\ub).
			\end{align*}
	\end{proof}
	Note that, in a robust MPC framework, this lemma amounts to assuming that the uncertainty cannot increase as additional information becomes available.
	Furthermore, a direct consequence of Assumption~\ref{a:unknown_constraints} is $\gb[n][k](\r^\x(t_n),\r^\u(t_n))\leq{}0\implies \gb[n][k+1](\r^\x(t_n),\r^\u(t_n))\leq{}0$. However, this does not entail Assumption~\ref{a:rec_ref-B}, which requires feasibility of the reference for all $n,k\in\mathbb{I}_0^\infty$. In many cases of interest, the reference does become infeasible at certain future times since the uncertainty can grow indefinitely large over time. This makes Assumption~\ref{a:rec_ref-B} false for a predefined trajectory. Some approaches tackle this issue by re-planning a feasible trajectory whenever some infeasibility is encountered~\cite{fulgenzi2008probabilistic}. However, only few such approaches can guarantee that a feasible trajectory always exists. 
	Our approach formalizes ideas similar to those of~\cite{petti2005safe,liu2017provably}, and always accounts for the worst-case scenario in order to guarantee that feasibility, hence safety, is never jeopardized even in case Assumption~\ref{a:rec_ref-B} does not hold. We ought to stress here that our approach can be deployed in combination with re-planning strategies in order to provide such guarantees.  
		
	\begin{Remark}
		In the context of autonomous driving, the constraints~$g_{n|k}$ could enforce avoiding collisions with obstacles (e.g., other road users) detected by the sensors, whose behavior can just be predicted, to some extent. Hence, Assumption~\ref{a:unknown_constraints} amounts to assuming that the uncertainty on, e.g., position and velocity of the detected objects at a specific time instance cannot increase as additional information becomes available. \review{We note however that limited sensor range makes it impossible to detect obstacles which are too far away. To ensure satisfaction of Assumption~\ref{a:unknown_constraints} one can adopt a worst-case approach which ensures that the predicted trajectory $\xb$ may never leave the sensor range, and by also assuming that new obstacles appear at the boundary of the sensor range at all times.}
	\end{Remark}
	
	\subsection{Terminal Conditions}
	\label{sec:terminal_safe_set}
	
	Assumption~\ref{a:terminal-B} poses difficulties, since the terminal controller $\kappa_\r^\mathrm{f}(\x,t)$ must satisfy $\gb(\xb,\kappa^\mathrm{f}_\r(\xb,\tb))\leq 0$, for all $n,k\in\mathbb{I}_0^\infty$, i.e., over an infinite horizon. However, this is in general impossible, unless additional assumptions are introduced, since the constraint uncertainty typically grows unbounded with time. 
		Furthermore, the terminal controller $\kappa_\r^\mathrm{f}(\x,t)$ is assumed to ensure stability with respect to the reference, i.e., $p_\r(\mathbf{x}_+^\kappa,t_+) - p_\r(\mathbf{x},t) \leq{} - q_\r(\mathbf{x},\kappa^\mathrm{f}_\r(\mathbf{x},t),t)$. This may not always be possible, since by dropping Assumption~\ref{a:rec_ref-B}, the reference might become infeasible with respect to some of the a-priori unknown constraints.

		 To cope with the recursive feasibility issue, we will assume the existence of a \emph{safe set} in which a-priori unknown constraints are guaranteed to be satisfied in all circumstances. This will allow us to rely on
		 standard approaches in MPC~\cite{borrelli2017predictive,kerrigan2001robust,yu2013tube} which are based on the existence of a robust invariant set. 
	
	\begin{Assumption}\label{a:safe}
		There exists a robust invariant set denoted $\mathcal{X}_\mathrm{safe}(\tb)\subseteq\mathbb{R}^{n_x}$ such that for all $\xb\in \mathcal{X}_\mathrm{safe}(\tb)$ there exists a safe control set $\mathcal{U}_\mathrm{safe}(\xb,\tb)\subseteq\mathbb{R}^{n_u+1}$ entailing that  $f(\xb,\u_\mathrm{safe})\in\mathcal{X}_\mathrm{safe}(\tb+t_\mathrm{s}+v_\mathrm{safe})$, and $h_n(\xb,\u_\mathrm{safe})\leq{}0$, for all $(\u_\mathrm{safe},\ v_\mathrm{safe})\in\mathcal{U}_\mathrm{safe}(\xb,\tb)$ and for all $n\geq{}k$. Moreover, for all $\xb\in \mathcal{X}_\mathrm{safe}(\tb)$ the a-priori unknown constraints can never be positive, i.e., by construction $\gb[n][k](\xb[n],\u_\mathrm{safe}) \leq 0$ for all $\xb\in \mathcal{X}_\mathrm{safe}(\tb)$ and $(\u_\mathrm{safe},\ v_\mathrm{safe})\in\mathcal{U}_\mathrm{safe}(\xb,\tb)$.
	\end{Assumption}
	\review{While this assumption might seem strong, it only postulates the existance of \emph{known} safe configurations for system~\eqref{eq:sys}. However, if no such configurations exist, then the controller based on Problem~\eqref{eq:nmpc} is intrinsically unsafe. On the other hand, if such configurations do exist for~\eqref{eq:sys}, then the safe set $\mathcal{X}_\mathrm{safe}$ is non-empty and invariant. Note that the safe configuration depends on system~\eqref{eq:sys}, the problem setting, and must be known a-priori.}
		
	\begin{Example}\label{ex:safe_set}
	\review{Many practical settings where safety is emphasized consider a system to be safe at  steady-state, in which case the safety set $\mathcal{X}_\mathrm{safe}$  can be formulated as}
	\review{\begin{align}
	\label{eq:safe_set_steady_state}
	\begin{split}
	\mathcal{X}_\mathrm{safe}(t_k):=\{ \, & \x \, | \, \x = f(\x,\u),\\
	&h_{k}(\x,\u) \leq 0, \, m_{k}(\x,\u) \leq 0 \, \},
	\end{split}
	\end{align}}
	where function $m_k$ defines additional constraints which might be needed in the set definition. Notable examples include, e.g., the following: (a) a robotic manipulator operating in a mixed human-robot environment is considered safe if it does not move; (b) a \review{vehicle parked in a safe configuration, e.g., a parking lot, emergency lane or any other safe environment that can be modeled by $m_k$}, is not responsible for collisions with other road users; (c) an electric circuit which is switched off is generally safe; (d) a ship docked in a port can be considered safe. We stress that the examples \review{together with set~\eqref{eq:safe_set_steady_state} presented above} are simplified in order to convey the safety message, but in practice the safe set will need to be carefully designed for each specific use case. 
	\end{Example}

	In general, most processes controlled either by humans or by automatic controllers do have emergency procedures which are triggered whenever safety is jeopardized. Assumption~\ref{a:safe} is meant to cover all these situations. 
	Note that this assumption entails that, when using the uncertainty model presented in Section~\ref{sec:coll_avoidance}, function $\gamma$ used in~\eqref{eq:constr_prediction} is not simply the output of the uncertainty model, but also includes the information that $\xb\in \mathcal{X}_\mathrm{safe}(\tb)$ implies $\gamma_{n|k}(\xb,\u_\mathrm{safe},\wb)=0$. This condition is typically not physics-driven, but stems from Assumption~\ref{a:safe}: we provide a clarifying example next.
		\begin{Example}
			Consider Example~\ref{ex:safe_set}, case (b), where functions $g$ and $\gamma$ model the set of positions that a pedestrian can reach. Whenever the vehicle is parked \review{in a safe configuration}, function $\gamma$ is not simply yielding all positions that a pedestrian can reach, but it must include the information that the pedestrian cannot reach the position of the vehicle, even though it would be physically possible to do so.
		\end{Example}
	
	\begin{Remark}
		Note that, while it is often reasonable to construct the set based on~\eqref{eq:safe_set_steady_state}, the safe set does not necessarily need to be forcing a steady-state, \review{other formulations, e.g., including safe periodic trajectories may also be considered.} Therefore, we prefer formulating the assumption in a generic way in order to cover as many cases of interest as possible.
	\end{Remark}
	
	The introduction of Assumption~\ref{a:safe}, allows us to drop Assumption~\ref{a:terminal-B}. We build our approach based on standard strategies in MPC~\cite{borrelli2017predictive,kerrigan2001robust,yu2013tube}, i.e., we rely on stabilizing terminal control laws $\kappa_\r^\mathrm{s}(\x,t)$ and sets $\mathcal{X}_\r^\mathrm{s}(t)$ satisfying Assumption~\ref{a:terminal-A}. 
		In order to obtain recursive feasibility also with respect to a-priori unknown constraints, we rely on the safe set $\mathcal{X}_\mathrm{safe}$ to introduce the following terminal set
	\begin{subequations}\label{eq:stabilizing_safe_set}
		\begin{align}
		\mathcal{X}^\mathrm{f}_\r(\tb[k+N]) &:=  \{\xb[k+N] \ | \ \exists \ \ub[n], \vb[n],  \label{eq:stab_a}\\
		&\tb[n+1]=\tb[n]+t_\mathrm{s} + \vb[n],\label{eq:stab_b}\\
		&\xb[n+1] =f(\xb[n],\ub[n]) ,  \\
		&\hb[n](\xb[n],\ub[n]) \leq{} 0,\\
		&\gb[n][k](\xb[n],\ub[n]) \leq{} 0,\\
		&\xb[n] \in \mathcal{X}^\mathrm{s}_\r(\tb), \label{eq:stab_f}\\
		&\xb[k+M]\in\mathcal{X}_\mathrm{safe}(\tb[k+M])\subseteq\mathcal{X}_\r^\mathrm{s}(\tb[k+M]),\label{eq:stab_g}\\
		& \eqref{eq:stab_a}-\eqref{eq:stab_f},\ \forall n\in \mathbb{I}_{k+N}^{k+M-1}\},\label{eq:stab_h}
		\end{align}
 	\end{subequations}
	where $M\geq{}N$ is a degree of freedom. Note that the construction of \eqref{eq:stabilizing_safe_set} implies that $\mathcal{X}_\mathrm{safe}(\tau)\subseteq{}\mathcal{X}_r^\mathrm{s}(\tau)$, $\forall\tau\geq{}0$. If \eqref{eq:stabilizing_safe_set} is a non-empty set we are guaranteed that for all $\x\in\mathcal{X}_\r^\mathrm{f}(\tau)$ a terminal control law exists,  which steers the states to the safe set.
	
	In order to provide a practical approach to design the terminal control law, we propose to first design a control law $\kappa_\mathrm{r}^\mathrm{s}$ as one would do in standard MPC formulations, i.e., by ignoring a-priori unknown constraints $g$ and by forcing the time in the reference to evolve according to its true dynamics. We can then define the terminal control law $(\kappa^\mathrm{f}_\r(\xb[k+N],\tb[k+N]),\nu_\r^\mathrm{f}(\xb[k+N],\tb[k+N]))$ by using $\kappa_\mathrm{r}^\mathrm{s}$, as the solution of
	\begin{subequations}\label{eq:terminal_controller_min}
	\begin{align}
		\min_{\u,\nu} \ \ & \|\u-\kappa_\mathrm{r}^\mathrm{s}(\xb[k+N],\tb[k+N])\|_2 + \nu^2 \\
		\mathrm{s.t.} \ \ &f(\xb[k+N],\u)\in\mathcal{X}_\r^\mathrm{f}(\tb[k+N]+t_\mathrm{s}+\nu), \\
		& h_{k+N}(\xb[k+N],\u) \leq0, \\
		& g_{k+N|k}(\xb[k+N],\u) \leq0.
		\end{align}
	\end{subequations}

	The idea behind the terminal set  \eqref{eq:stabilizing_safe_set} is to ensure safety by forcing the system to be able to reach a safe set $\mathcal{X}_\mathrm{safe}$ in a finite amount of time $M-N\geq{}0$, while always remaining inside a stabilizing set $\mathcal{X}_\mathrm{r}^\mathrm{s}(t)$ around the reference. An illustrative example is shown in Section \ref{sec:example_double}. Note that $M$ is a parameter which can be used to tune the stabilizing terminal safe set and, consequently, the NMPC scheme~\eqref{eq:nmpc}. If $M=N$, then the terminal set coincides with the safe set, possibly limiting the  capabilities of the terminal control law, i.e., $\kappa_\r^\mathrm{f}(\x,\tau)\neq\kappa_\r^\mathrm{s}(\x,\tau)$ and $\nu_\r^\mathrm{f}(\x,\tau)\neq{}0$. On the other hand if $M \gg N$, the computational complexity of $\mathcal{X}_\r^\mathrm{f}$ can become excessive.

	We can now prove recursive feasibility for the terminal control law.
	\begin{Lemma}\label{lemma:1}
		Suppose that Assumptions \ref{a:cont}, \ref{a:rec_ref-A}, \ref{a:terminal-A}, \ref{a:unknown_constraints}, and \ref{a:safe} hold, and that the terminal set $\mathcal{X}_\mathrm{r}^\mathrm{f}$ given by \eqref{eq:stabilizing_safe_set} is nonempty. Then, for an initial state $\xb[k+N]\in\mathcal{X}_\r^\mathrm{f}(\tb[k+N])$, the terminal controller \eqref{eq:terminal_controller_min} is recursively feasible.
		
		\begin{proof}
			Consider the state $(\xb[k+N],\tb[k+N])$ at time $k$ such that $\xb[k+N]\in\mathcal{X}_\mathrm{r}^\mathrm{f}(\tb[k+N])$. By construction, there exist control sequences $\mathbf{U}_k=\{\ub[k+N],...,\ub[k+M-1]\}$ and $\mathbf{V}_k=\{\vb[k+M],...,\vb[k+M-1]\}$ that generate the corresponding state trajectories $\mathbf{X}_k=\{\xb[k+N],...,\xb[k+M]\}$ and $\mathbf{T}_k=\{\tb[k+N],...,\tb[k+M]\}$ satisfying \eqref{eq:stabilizing_safe_set}. Hence, we have that $h_{k+N}(\xb[k+N],\ub[k+N])\leq{}0$ and $\gb[k+N](\xb[k+N],\ub[k+N])\leq{}0$. By applying the first control inputs $\ub[k+N]$ and $\vb[k+N]$, we know that 
			the following trajectories
			\begin{align}
			\mathbf{X}_{k+1} &= \{\xb[k+N+1],...,\xb[k+M],f(\xb[k+M],\u_\mathrm{safe}) \},\\
			\mathbf{T}_{k+1} &= \{\tb[k+N+1],...,\tb[k+M], \tb[k+M] + t_\mathrm{s} + v_\mathrm{safe} \},\\
			\mathbf{U}_{k+1} &= \{\ub[k+N+1],...,\ub[k+M-1],\u_\mathrm{safe} \},\\
			\mathbf{V}_{k+1} &= \{\vb[k+N+1],...,\vb[k+M-1],\vv_\mathrm{safe} \},
			\end{align}
			satisfy $h_n(\xb,\ub)\leq{}0$ for $n\in\mathbb{I}_{k+N+1}^{k+M-1}$ at time $k+1$ by definition, while Assumption~\ref{a:unknown_constraints} ensures that
				$$\gb[n][k+1](\xb,\ub)\leq{}\gb[n][k](\xb,\ub)\leq{}0,\ \forall n\in\mathbb{I}_{k+N}^{k+M-1}.$$
			Furthermore, since $\xb[k+M]\in\mathcal{X}_\mathrm{safe}(\tb[k+M])$, Assumption~\ref{a:safe} ensures $\hb[k+M](\xb,\u_\mathrm{safe})\leq{}0$, $\gb[k+M][k+1](\xb[k+M],\u_\mathrm{safe})\leq{}0$, and that $f(\xb[k+M],\u_\mathrm{safe})\in\mathcal{X}_\mathrm{safe}$. Hence, we have shown that there exists control inputs $\u:=\ub[k+N]$ and $v:=\vb[k+N]$ ensuring $f(\xb[k+N],\u)\in\mathcal{X}_\r^\mathrm{f}(\tb[k+N]+t_\mathrm{s}+v)$ and that Problem~\eqref{eq:terminal_controller_min} is feasible at time $k$. To prove feasibility for times $\kappa^\prime\geq{}k+1$, we can use the invariance condition of Assumption~\ref{a:safe} and construct similar guesses.
		\end{proof}	
	\end{Lemma}

	\subsection{MPC Recursive Feasibility in Uncertain Environments}\label{sec:theorem}

	Building on the results presented in Sections~\ref{sec:coll_avoidance} and ~\ref{sec:terminal_safe_set}, we can prove recursive feasibility for the proposed MPC scheme.
	
	\begin{Theorem}[Recursive Feasibility]\label{theorem:1} Suppose that Assumptions \ref{a:cont}, \ref{a:rec_ref-A}, \ref{a:terminal-A}, \ref{a:unknown_constraints}, and \ref{a:safe} hold, and that Problem~\eqref{eq:nmpc} is feasible for the initial state $(\x_k,\tau_k)$, with terminal set and terminal controllers given by \eqref{eq:stabilizing_safe_set} and \eqref{eq:terminal_controller_min}, respectively. Then, system \eqref{eq:sys}-\eqref{eq:tau_controlled} in closed loop with the solution of~\eqref{eq:nmpc} applied in receding horizon is safe (recursively feasible) at all times.
		\begin{proof}
			Consider the following control trajectories $\mathbf{U}^\star_k=\{\ub[k]^\star,...,\ub[k+N-1]^\star\}$ and $\mathbf{V}^\star_k=\{\vb[k]^\star,...,\vb[k+N-1]^\star\}$, with corresponding state trajectories $\mathbf{X}_k^\star = \{\xb[k]^\star,...,\xb[k+N]^\star\}$ and $\mathbf{T}_k^\star = \{\tb[k]^\star,...,\tb[k+N]^\star\}$, to be the solution of Problem~\eqref{eq:nmpc}. Recursive feasibility follows from Assumption~\ref{a:unknown_constraints} and Lemma~\ref{lemma:1}, which ensure that the control and state trajectories $\mathbf{U}^\star_k$, $\mathbf{V}^\star_k$, $\mathbf{X}_k^\star$, and $\mathbf{T}_k^\star$, and their prolongation to infinite time using $\kappa_\r^\mathrm{f}(\xb[k^\prime],\tb[k^\prime])$ and $\vv_\r^\mathrm{f}(\xb[k^\prime],\tb[k^\prime])$, satisfy constraints $g_{k^{\prime}|k^{\prime\prime}}$ for all $k^{\prime}\geq{}k^{\prime\prime}$ and $k^{\prime\prime}\geq{}k$. This, in turn, implies that the controller is safe in the sense of Definition \ref{def:safe}.
		\end{proof}
	\end{Theorem}

	While we have proven recursive feasibility, the presence of obstacles makes it harder to discuss closed-loop stability. In principle it seems possible to prove Input-To-State Stability~(ISS), and we will investigate this possibility in future research. Here, we limit ourselves to the observation that, if the a-priori unknown constraints become inactive, then the proposed formulation yields nominal asymptotic stability, which can be proven by standard arguments such as those used in Theorem~\ref{prop:stab_feas}.

	\section{Simulations}\label{sec:simulations}
	In this section we propose three examples to illustrate the developed theory. The first example aims at discussing how the novel MPFTC framework works, and how it compares to the similar and competing MPFC framework using the example in~\cite{Faulwasser2009}. The second example focuses the recursive feasibility aspect. Here a very simple double integrator is used to in order to illustrate all important aspects, and to be able to visualize all features. Finally, the third example, which was also used in~\cite{Faulwasser2016}, shows how our safety enforcing formulation performs for a more involved setting.
	
	In all the examples we will use stage and terminal cost in \eqref{eq:stage_cost}-\eqref{eq:terminal_cost}, i.e.,
	\begin{gather*}
	\Delta\xb := \xb-\rx(\tb),\quad \Delta\ub:=\ub-\ru(\tb), \nonumber\\
	q_\r(\xb,\ub,\tb) = \matr{c}{\Delta\xb\\\Delta\ub}^\top{}W\matr{c}{\Delta\xb\\\Delta\ub},\\
	p_\r(\xb[k+N],\tb[k+N]) = \Delta\xb[k+N]^\top{}P\Delta\xb[k+N].
	\end{gather*}
	
	Furthermore, all simulations ran on a laptop computer (i7 2.8GHz, 16GB RAM) and were implemented in Matlab using the CasADi~\cite{andersson2019casadi} software together with the IPOPT~\cite{wachter2006implementation} solver. \review{The simulation scripts have been made available on Github\footnote{{\tt \review{github.com/ivobatkovic/safe-trajectory-tracking}}}.}
	
	\subsection{Practical Safe MPFTC Formulation}
	Since computing terminal set~\eqref{eq:stabilizing_safe_set} explicitly is impractical for non-trivial cases, 
	we extend the MPC prediction horizon from $N$ to $M$ and include the conditions defining~\eqref{eq:stabilizing_safe_set} as constraints in the MPC problem. More formally, for practical applications, we propose to use the following problem formulation, which is equivalent to Problem~\eqref{eq:nmpc}
	\begin{subequations}\label{eq:practical}
		\begin{align}
		\min_{\substack{\bx\\\btau},\substack{\bu\\\bv}}& \sum_{n=k}^{k+N-1}
		q_\r(\xb,\ub,\tb)+w \vb^2&\\ 
		& &\hspace{-10em}+p_\r(\xb[k+N],\tb[k+N])\nonumber\\
		\text{s.t.}\ &\xb[k][k] = \x_{k},\  \tb[k] = \tau_{k},& \\
		&\xb[n+1] = f(\xb,\ub), &\hspace{-1em}n\in \mathbb{I}_k^{k+M-1},\\
		&\tb[n+1] = \tb+t_\mathrm{s}+\vb, &\hspace{-1em}n\in \mathbb{I}_k^{k+M-1},\\
		&\hb(\xb,\ub) \leq{} 0, & \hspace{-1em}n\in \mathbb{I}_k^{k+M-1},\\
		&\gb[n][k](\xb,\ub) \leq{} 0, & \hspace{-1em}n\in \mathbb{I}_k^{k+M-1},\\
		&\xb[k+n] \in\mathcal{X}_\r^\mathrm{s}(\tb[k+n]),&\hspace{-1em}n\in \mathbb{I}_{k+N}^{k+M-1},\\
		&\xb[k+M]\in\mathcal{X}_\mathrm{safe}(\tb[k+M])\subseteq\mathcal{X}_\r^\mathrm{s}(\tb[k+M]).\hspace{-10em}&
		\end{align}
	\end{subequations}
	This scheme can be seen as a 2-stage problem, where the first stage ($n\in\mathbb{I}_{k}^{k+N}$) defines the MPC problem and the second stage  ($n\in\mathbb{I}_{k+N+1}^{k+M}$) defines the terminal set implicitly. While this formulation solves the issue of precomputing $\mathcal{{X}}_\r^\mathrm{f}$ explicitly, it can still suffer from numerical difficulties since the trajectory in the second stage ($n\in\mathbb{I}_{k+N+1}^{k+M}$) is not penalized by any cost. While a thorough discussion on possible remedies is out of the scope of this paper, we observe the following: (a) for interior-point methods, the primal interpretation results in a cost penalizing deviations from the center of the feasible domain, such that the solution is unique and well-defined; (b) for any iterative solver, introducing a penalization on deviations from the previous iterate has a regularizing effect which alleviates the numerical difficulties; (c) tracking the reference also beyond $n=N$ with a small penalty function introduces a perturbation on the optimal solution which can be analyzed in the context of ISS and is expected to result in a small tracking inaccuracy which could be acceptable in many practical situations.
	\begin{figure*}[ht]
		\centering
		\includegraphics[width=\linewidth]{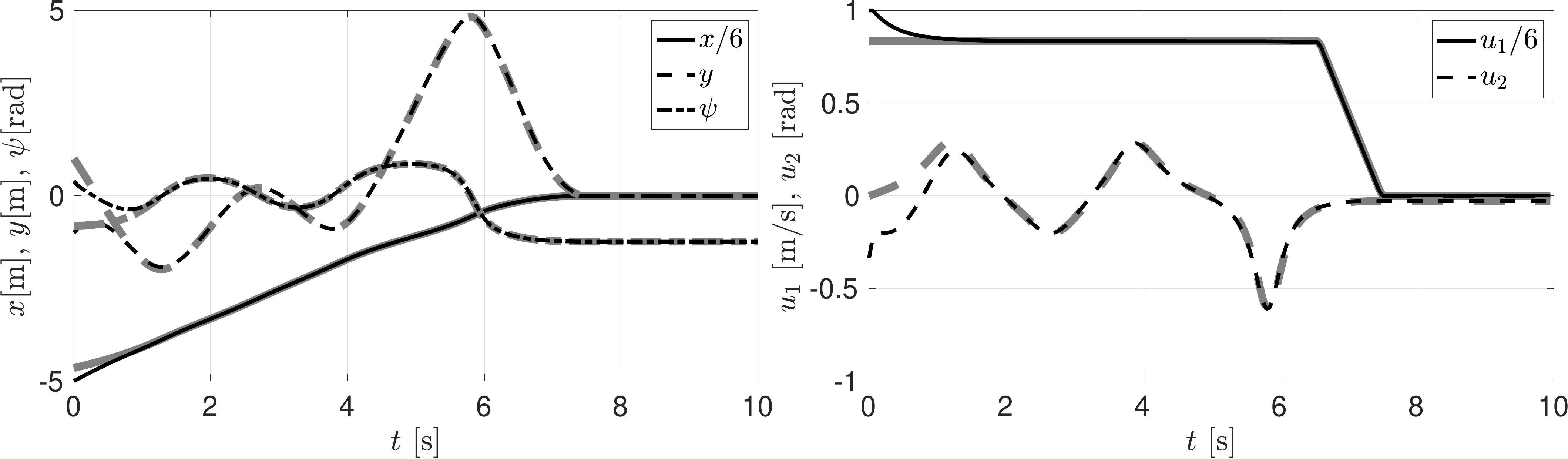}  
		\caption{Closed loop trajectories for the autonomous vehicle example with MPFTC with $w=10$. Gray lines with matching line style denote the reference trajectory.}
		\label{fig:vehicle_mpftc_states}
	\end{figure*}

	\begin{figure}[ht]
		\mbox{\parbox{.48\textwidth}{
				\centering
				\includegraphics[width=\linewidth]{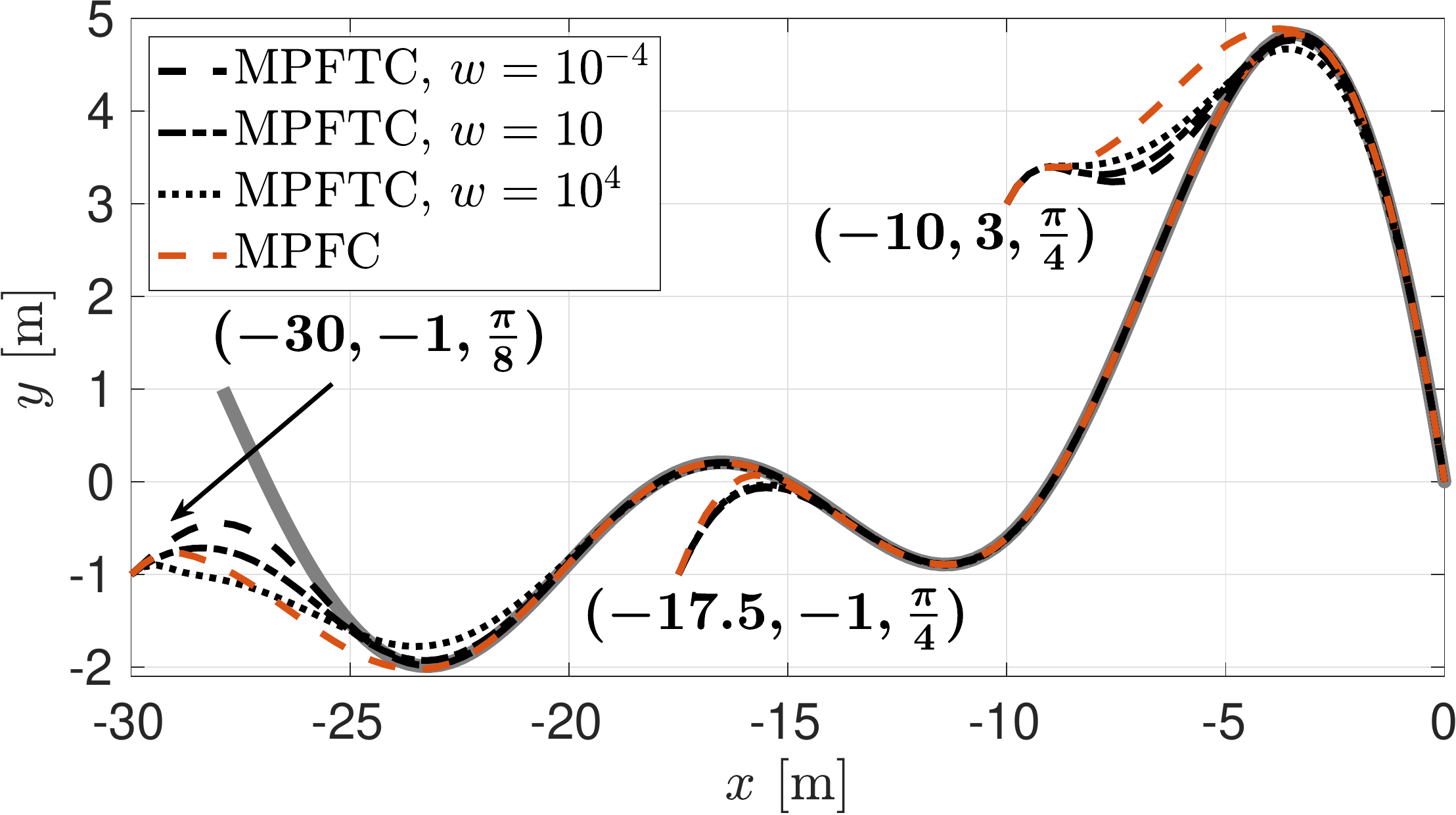}  
		}}
		\caption{Closed loop trajectories of the $x$-$y$ states for the autonomous vehicle example.}
		\label{fig:vehicle_mpftc_closed}
	\end{figure}
	
	\subsection{Autonomous Vehicle Reference Tracking}
	We evaluate MPFTC and we compare the results with the MPFC formulation of~\cite{Faulwasser2009} using the kinematic model of a car
	\begin{equation}\label{eq:robot_system}
	\matr{c}{\dot{x}\\\dot{y}\\\dot\psi} = \matr{c}{u_1\cos\psi\\u_1\sin\psi\\u_1\tan{}u_2},
	\end{equation}
	where $x$ and $y$ describe the position, $\psi$ is the yaw angle, $u_1$ is the speed, and $u_2$ the steering angle. The system is subject to the constraints $0\leq{}u_1\leq{}6,\,\|u_2\|_2\leq{}0.63$.
	As reference, we use the path 
	\begin{equation}
	p(\theta) = \left [\rho(\theta)^\top,\arctan\left (\frac{\partial\rho_2}{\partial\theta}\right )\right ]^\top,
	\end{equation}
	\review{with $\rho(\theta)=\left[ \theta,-6\log(20/(5+|	\theta|))\sin(0.35\theta)\right ]^\top$} and parameter $\theta\in[-30,0]$. 
	Since MPTFC requires a trajectory to track, we design the a-priori path evolution of $\theta$ to be
	\begin{align*}
	\dot{\theta}(t) &= \frac{v_\mathrm{ref}(t) }{\left \| \nabla_\theta \rho(\theta(t))\right \|_2}, && v_\mathrm{ref}(t) = \left \{ 
	\begin{array}{@{}ll@{}}
	\phantom{-}5 & t\leq{}7\\
	\max(5-a t,0) & t>7
	\end{array}
	\right . ,
	\end{align*}
	with $a=-5.38$. This predefined path evolution implies that the  reference velocity along the trajectory will be $5 \ \mathrm{m/s}$ for $t\leq{}7$, and monotonically decreasing for $t>7$ until reaching zero velocity.
	For MPFTC~\eqref{eq:nmpc}, with auxiliary state $\tau$ and control $v$ we formulate a feasible reference as
	\begin{gather}
	\r^\x(t) = \matr{c}{\rho(\theta(t))^\top,\, \arctan\left (\frac{\partial\rho_2}{\partial\theta}\right ) }^\top,\\
	\r^\u(t) = \matr{c}{\dot{\theta}(t)\sqrt{1+\left (\frac{\partial\rho_1(\theta(t))}{\partial\theta}\right )^2}\\
		\arctan\left (\frac{\partial^2\rho_2}{\partial\theta^2}\left (1+\left (\frac{\partial\rho_2}{\partial\theta}\right )^2\right )^{-\frac{3}{2}}\right )
	},
	\end{gather}
	where the control reference can directly be derived from \eqref{eq:robot_system}.
	For the cost we use $W = \mathrm{blockdiag}(Q,R)$ with
	\begin{gather*}\label{eq:tuning}
	Q=\mathrm{diag}(1,1,1),\quad
	R=\mathrm{diag}(1,1),\quad w=10.
	\end{gather*}
	Using the results from (\cite{Faulwasser2009}, Collorary 1) it can be shown that the terminal cost given by $P=Q$, together with the terminal set
	$\mathcal{X}^\mathrm{f}_\r(t) =\{ \x\, |\, \x=\r^x(t)\}$, is a suitable, although conservative, choice to stabilize the system. 
	
	For MPFC, we use the tuning parameters $Q_\mathrm{MPFC}=8\cdot{}\mathrm{diag}(10^4,10^5,10^5,1/16)$, $R_\mathrm{MPFC}=\mathrm{diag}(10,10,1)$, terminal weight $Q_{\mathrm{MPFC},N}=1/2\cdot{}\mathrm{diag}(0,0,0,1740)$, and same setup as in \cite{Faulwasser2009}.
	
	For the simulation, we use a prediction horizon of $1$s, and solve the OCP repeatedly with 20 direct multiple shooting intervals. The closed loop sampling time is \review{$\delta=0.05$s}. The initial value for $\tau_0$ is selected by projecting the initial position $(x_0,y_0)$ on the trajectory, i.e., $$\tau_0=\arg\min\|(x_0,y_0)-\rho(\theta(\tau))\|_2.$$

	Figure \ref{fig:vehicle_mpftc_states} shows the closed-loop trajectories for the initial value $\x_0=[-30,-1,\pi/8]^\top$, $\tau_0=0.573$s for the MPFTC controller with $w=10$. Starting with an offset from the reference trajectory, for all initial states the system is stabilized to the reference and converges to the point $p(0)$. Figure \ref{fig:vehicle_mpftc_closed} shows closed-loop trajectories for different initial values for both MPFTC and MPFC. For MPFTC, we show how the closed-loop behavior depends on the auxiliary input cost $w$: higher penalties result in a more aggressive tracking than lower ones. Note that while MPFTC and MPFC result in different closed-loop trajectories, their behavior is similar. For the simulations with initial condition $\x_0=[-30,-1,\pi/8]^\top$, we obtained the average runtime of $\hat{t}_\mathrm{MPFTC}=0.049$s and $\hat{t}_\mathrm{MPFC}=0.044$s for MPFTC and MPFC, respectively. We stress our implementation was far from being implemented in a computationally optimal way, such that a real-time implementation is expected to provide much faster run times.
	
	\subsection{Double Integrator with A-Priori Unknown Obstacles}\label{sec:example_double}
	We consider double integrator dynamics to illustrate in the simplest fashion the effect of the safe set formulation in the presence of a-priori unknown obstacles. The state and control are $\x=[p,\dot{p}]^\top$, $\dot{p}\geq{}0$ and $\u=a\in[-1,5]$ respectively, with
	reference $\r^\x(t)=[tv^{\mathrm{r}},v^{\mathrm{r}}]^\top$, $\r^\u(t)=0$. For the cost we use $W = \mathrm{blockdiag}(Q,R)$ with
	\begin{align*}
	Q=\mathrm{diag}(10,10), && R=1, && w=1.
	\end{align*}
	The terminal cost matrix $P$ is obtained from the LQR cost corresponding to $Q_\mathrm{LQR}=\mathrm{diag}(1,1)$, $R_\mathrm{LQR}=10$, and 
	a corresponding stabilizing set for the LQR controller $\mathcal{X}^\mathrm{s}_\r(t)=\{ \x\,|\, -K(\x-\r^\x(t))\in[-1,5]\}$. We construct the safe set using \eqref{eq:safe_set_steady_state} to obtain $\mathcal{X}_\mathrm{safe}(t)=\{ \x\,|\,\dot{p}=0\}\cap\mathcal{X}_\r^\mathrm{s}(t)$. With $\mathcal{X}^\mathrm{s}_\r(\tb[k+n])$ and $\mathcal{{X}}_\mathrm{safe}$, we use the terminal set $\mathcal{X}^\mathrm{f}_\r(\tb[k+N])$  given by \eqref{eq:stabilizing_safe_set}. 
	
	\begin{figure*}[ht!]
		\centering
		\includegraphics[width=\linewidth]{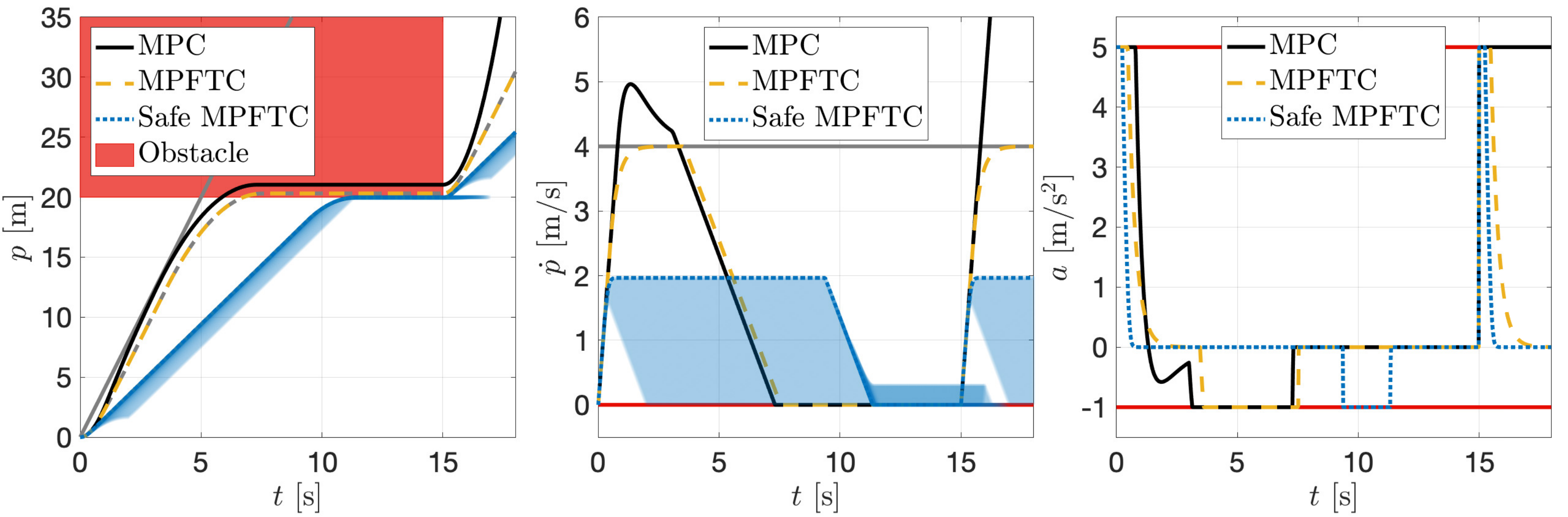}  
		\caption{Closed-loop trajectories for the double integrator example. The red box and lines denote constraints, while the gray lines denote the reference trajectory, with the line style matching the corresponding state trajectory. The opaque blue lines in the two left-most figures show the safe MPFTC open-loop trajectories at all times for the position and velocity respectively.}
		\label{fig:safe_mpftc_crash}
	\end{figure*}
	
	We introduce a static obstacle at position $p^{\mathrm{obs}}$, such that
	\begin{equation}\label{eq:double_integrator_obstacle}
	\gb(\xb,\ub) = p_{n|k} - p^{\mathrm{obs}}_{n|k} \leq{} 0.
	\end{equation}
	We position the obstacle at $p^{\mathrm{obs}}=20$m at times $t\in[0,15]$. For the following problems not including  the safety formulation, we relax the constraint with an exact penalty \cite{Scokaert1999a}.

	For the following scenario we compare MPFTC with standard MPC, i.e., without reference adaptation.
	To highlight the benefit of a safety set, we select a reference $v^\mathrm{r}=4$m/s with a sampling time of $t_s=0.02$s and control intervals $N=50$, and $M=100$. For MPFTC and MPC without the stabilizing safe set, we use the same sampling time but set $N=100$. 
	
	Figure \ref{fig:safe_mpftc_crash} shows that MPC and MPFTC without the proposed safety formulation are not able to satisfy the constraint. Additionally, MPC is very aggressive, especially after the obstacle is removed: while the system was stopped by the obstacle the gap in the position reference kept increasing, resulting in a wind-up effect. On the contrary, safe MPFTC satisfies the constraint and does not have an aggressive behavior. Moreover, it never attains the reference velocity of 4m/s since it would not be safe to do so. With a longer prediction horizon a velocity of 4m/s would be safe and in that case also safe MPFTC would reach the reference. Note that between times $t\in[12,15]s$, the state $\x_k=[p_k,\dot{p}_k]^\top$ is in fact in the safe set, since the velocity is  zero, i.e., $\dot{p}_k=0$. Then, after $15$s, the constraint is lifted and the controller continues to track the reference same as before. Finally, the blue opaque open loop state trajectories show how the safe set forces the velocity to be zero at the end of the horizon. For each simulation, we obtained the following average runtimes: $\hat{t}_\mathrm{MPC}=0.074$s, $\hat{t}_\mathrm{MPFTC}=0.046$s, and $\hat{t}^\mathrm{Safe}_\mathrm{MPFTC}=0.055$s, for MPC, MPFTC and safe MPFTC, respectively. Also in this case we stress that these run times only provide an indication and a real-time implementation is expected to be much faster.
	
	Since in general it is impractical or even impossible to compute the terminal set~\eqref{eq:stabilizing_safe_set} explicitly, we formulated safe MPFTC implicitly using the formulation proposed in Problem~\eqref{eq:practical}. However, in this simple example, such a set can be computed explicitly \review{using~\cite{MPT3}}. Figure \ref{fig:double_integrator_safe_terminal_set} shows the terminal set around the reference $\r^\x(t)$ for different values  of $\tb[k+N]$. When the reference is far from the obstacle,  the set is defined by the actuator limitations and the requirement to reach $\mathcal{X}_\mathrm{safe}$ in a finite number of time steps. However, as $\tb[k+N]$ approaches $5$ s, the reference approaches the obstacle $p^\mathrm{obs}$ and the set contracts in order to satisfy constraint \eqref{eq:double_integrator_obstacle}.

	It is well known that in optimal control, constraints appearing in the far future have a negligible impact on the initial control. This fact is tacitly exploited together with the exact penalty constraint relaxation in order to avoid feasibility issues. However, our formulation provides a rigorous approach and is particularly useful in cases in which the obstacles cannot be detected well in advance.
	
	Finally, in this example the constraints are not time-varying for the sake of simplicity. We provide next a more involved example.
	
	\begin{figure}[t]
		\mbox{\parbox{.48\textwidth}{
				\centering
				\includegraphics[width=\linewidth]{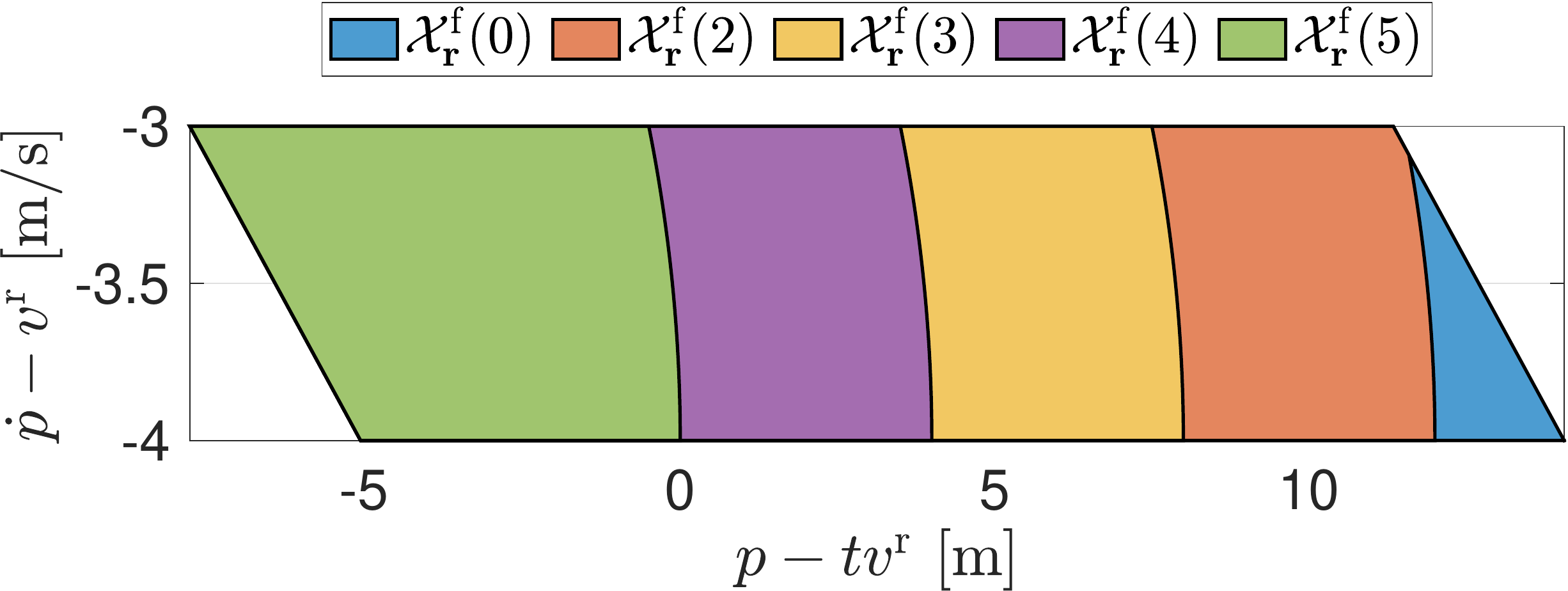}  
		}}
		\caption{Parametric stabilizing terminal safe set for the double integrator example. Note that $\mathcal{X}_\r^\mathrm{f}(\tau_+)\subseteq{}\mathcal{X}_\r^\mathrm{f}(\tau)$ for $\tau\leq \tau_+\leq 5$ s.}
		\label{fig:double_integrator_safe_terminal_set}
	\end{figure}

	\begin{figure*}[ht]
		\centering
		\includegraphics[width=\linewidth]{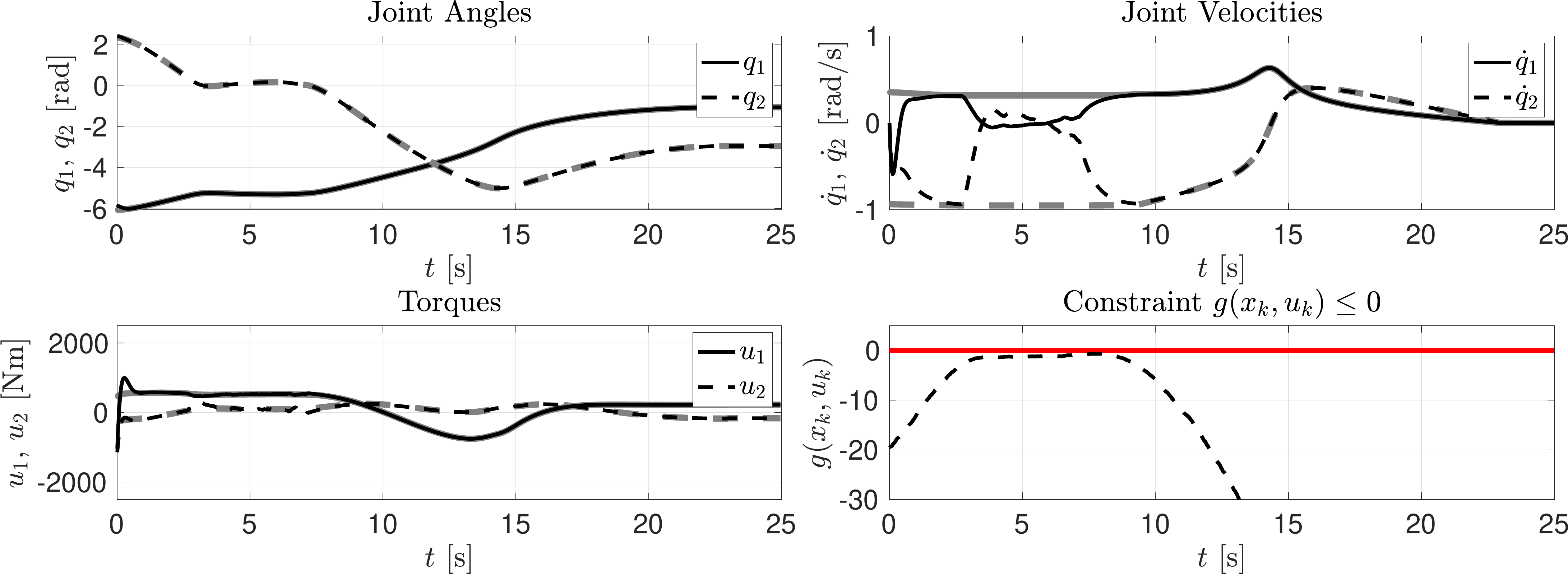}  
		\caption{Closed-loop trajectories for the robotic joint example. The gray lines denote the reference trajectory for each state and control, while the red lines denote constraint limits.}
		\label{fig:mpatc_1_states}
	\end{figure*}
	
	\begin{figure}[t]
		\centering
		\mbox{\parbox{.48\textwidth}{
				\centering
				\includegraphics[width=\linewidth]{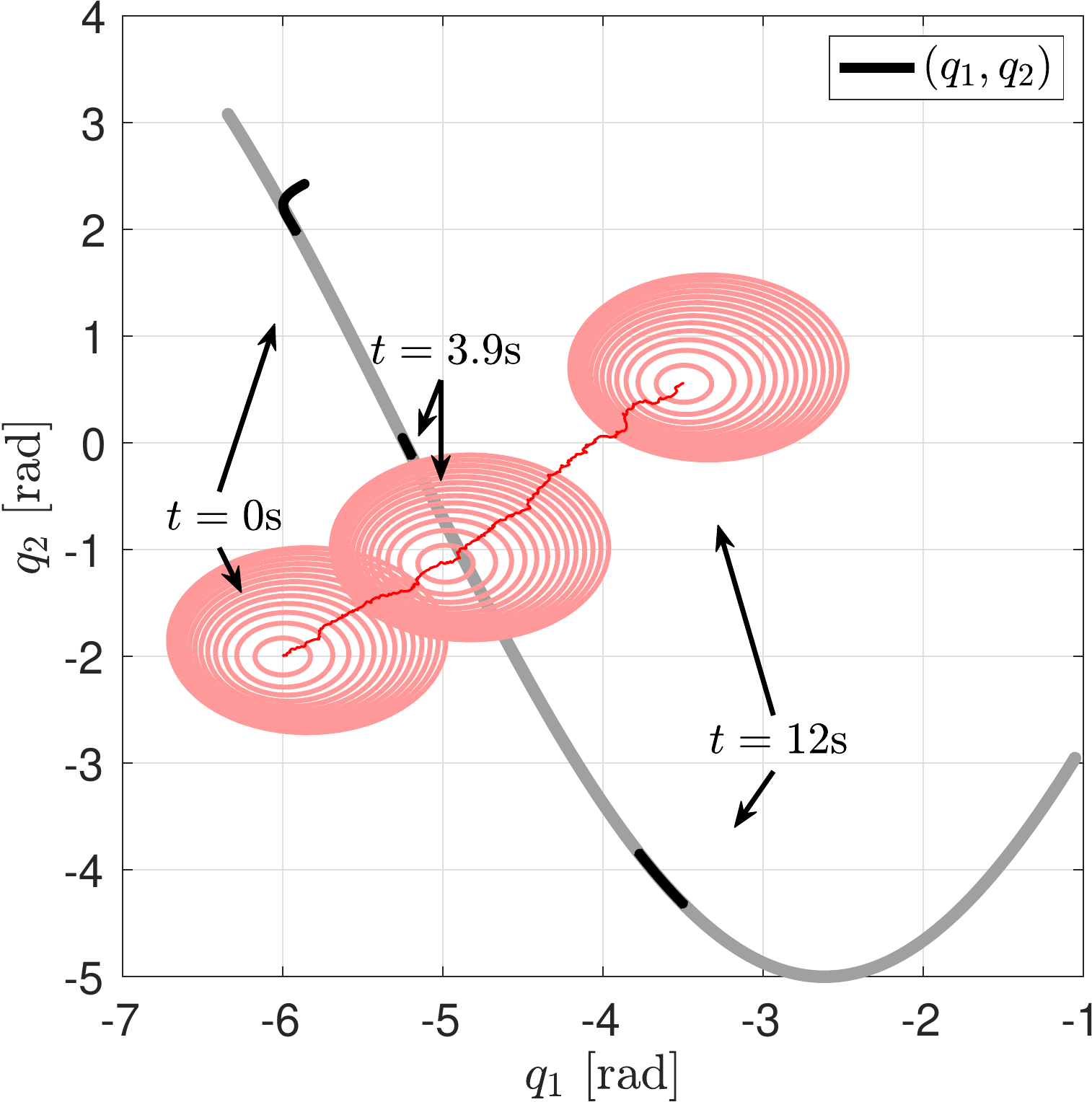}
		}}
		\caption{Open-loop predictions for the robotic joint example. The gray line denotes the reference, while the red circles show the growing uncertainty for different times. The red noisy line shows the closed-loop evolution of the uncertainty.}
		\label{fig:mpatc_1_q1q2}
	\end{figure}
	
	\subsection{Robotic Arm: Flexible Tracking with Obstacles}
	We consider a fully actuated planar robot with two degrees of freedom, no friction, nor external forces, and dynamics
	\begin{align}
	\matr{c}{\dot{x}_1\\\dot{x}_2} &= \matr{c}{ x_2\\B^{-1}(x_1)(u-C(x_1,x_2)x_2-g(x_1))},\label{eq:robot}
	\end{align}
	where $x_1=(q_1,q_2)$ are the joint angles, and $x_2=(\dot{q}_1,\dot{q}_2)$ the joint velocities. The full model description and the parameter values are given in Appendix \ref{appendix:model}.
	We consider the following box constraints on the state and control 
	\begin{align}\label{eq:box_constr}
	\|u\|_\infty\leq{}\bar{u}, && \|x_2\|_\infty\leq{}\bar{\dot{q}},
	\end{align}
	with $\bar{u}=4000$Nm and $\bar{\dot{q}}=(3/2)\pi$ rad/s.
	We consider
	\begin{equation}\label{eq:path}
	p(\theta)=\left (\theta-\frac{\pi}{3},\,5\sin\left (0.6 \left (\theta-\frac{\pi}{3}\right )\right )\right ),
	\end{equation}
	with $\theta\in[-5.3,0]$, as the desired path to be tracked and define the timing law 
	\begin{align*}
	\dot{\theta}(t) &= \frac{v_\mathrm{ref}(t) }{\left \| \nabla_\theta p(\theta(t))\right \|_2}, && v_\mathrm{ref}(t) = \left \{ 
	\begin{array}{@{}ll@{}}
	\phantom{-}1 & t\leq{}5\\
	\max(1-a t,0) & t>5
	\end{array}
	\right . ,
	\end{align*}
	with $a=-0.0734$. This predefined path evolution implies that the norm of the reference trajectory for the joint velocities will be  $1\ \mathrm{rad/s}$ for $t\leq{}5$, and monotonically decreasing for $t>5$ until reaching zero velocity.
	
	The state and input reference trajectories are given by
	\begin{align*}
	\r^\x(t) &= \matr{cc}{p(\theta(t))  & \frac{\partial{p}}{\partial\theta}\dot{\theta}(t)}^\top,\\
	\r^\u(t) &= \matr{c}{ B(x_{1})\ddot{p}(\theta(t))+C(x_1,x_2)x_{2}+g(x_1)}^\top,
	\end{align*}
	where the control reference follows from \eqref{eq:robot}. 
	For the cost we use $W = \mathrm{blockdiag}(Q,R)$ with
	\begin{align*}
	Q=\mathrm{diag}(10^5,10^5,10,10),\
	R=\mathrm{diag}(10^{-3},10^{-3}),\ 	w=10.
	\end{align*}
	The terminal cost matrix is given by $P=P_\eta$, and corresponding stabilizing set
	\begin{equation}\label{eq:ex_2_terminal}
	\mathcal{X}^\mathrm{s}_\r(\tb[k+n]) =\{ \xb[k+N]\, |\, \Delta\xb[k+N]^\top P_\eta\Delta\xb[k+N] \leq{} \gamma_*\},
	\end{equation}
	where the values and derivation of $P_\eta$ and $\gamma_*$ are given in Appendix \ref{appendix:terminal}. The safety set is	 constructed as
	\begin{align}
	\label{eq:robot_safe_set}
	\mathcal{X}_\mathrm{safe}(t)=\{(\dot{q}_1,\dot{q}_2)\,|\, \dot{q}_1=0,\dot{q}_2=0\} \cap \mathcal{X}_\r^\mathrm{s}(t),
	\end{align} and we use the terminal set $\mathcal{X}^\mathrm{f}_\r(\tb[k+N])$ given by \eqref{eq:stabilizing_safe_set} using~\eqref{eq:ex_2_terminal} and \eqref{eq:robot_safe_set}. 
	Since the safe set ensures a steady state, a minimal feasibility condition is that $\xb[k+N] = (p(\theta(\tb[k+N])),0,0)\in\mathcal{X}^\mathrm{f}_\r(\tb[k+N])$.
	
	We introduce the following time-varying uncertainty
	\begin{align}\label{eq:uncertain_model}
	\w_{k+1}= \omega(\w_{k},\xi) = \w_k + \matr{c}{0.3t_s\cos(\pi/4)+\xi_1\\0.3t_s\sin(\pi/4)+\xi_2},
	\end{align}
	where $\w\in\mathbb{R}^2$ is the uncertainty state and $\|(\xi_1,\xi_2)\|_2\leq{}0.03$ are bounded noise inputs. We construct the robust constraint under Assumption \ref{a:unknown_constraints} by propagating the expectation of \eqref{eq:uncertain_model}, and encapsulate the uncertainty with a circle of increasing radius $\delta{}\hat{r}=0.03$ and an initial uncertainty of $\hat{r}_0=0.03$. The constraint can then be expressed as
	\begin{align*}
	\gb[n][k](\xb,\ub) = \max_{\wb\in\Wb_{n|k}} \ \gamma_{n|k}(\xb,\ub,\wb)\leq{} 0,
	\end{align*}
	with 
	\begin{equation*}
	\gamma_{n|k}(\xb,\ub,\wb)=\hat{r}_{n|k}^2- (\xb[n][k]^{1,2}-\wb)^\top(\xb^{1,2}-\wb),
	\end{equation*}
	where $\xb^{1,2}=({q}_{1,n|k},{q}_{2,n|k})$ and $\hat{r}_{n|k}=\hat{r}_0+(n-k)\delta\hat{r}$. Generally, one could include a more complex environment uncertainty and prediction model, however, here we choose to use simple models to highlight the performance of our framework in a clear manner. 
	
	For the simulation, we set the sampling time to $0.03$s and use control intervals $N=25$, $M=50$, and $\w_0=(-6,-2)$ as initial value for the uncertainty. The initial value for $\tau_0$ is selected by projecting $(q_1,q_2)$ on the trajectory, i.e., $$\tau_0=\arg\min\|(q_1,q_2)-p(\theta(\tau))\|_2.$$

	Figure \ref{fig:mpatc_1_states} shows the closed-loop trajectories for the initial condition $(x_1,x_2)=(-5.86,2.43,0,0)$, $\tau=0.79$. The system is quickly stabilized to the reference; after $3$s the system deviates from the desired velocity reference due to the presence of an uncertainty. Tracking is temporarily lost for the velocities $(\dot{q}_1,\dot{q}_2)$ but to a lesser extent for the positions $(q_1,q_2)$. When the obstacle leaves, the system is again stabilized towards the reference. Figure \ref{fig:mpatc_1_q1q2} shows open-loop trajectories of the joint positions together with the moving obstacle at different times. The red circles represent the predicted uncertainty as a region in the joint space to be avoided. At times $t=0$s and $t=12$s the robotic arm can move freely, while for $t=3.9$s the robotic arm is forced to come close to a stop, and even reverse, in order to avoid the obstacle before being allowed to continue along the reference trajectory. \review{We observe that the simulation results remain indistinguishable if the initial auxiliary state $\tb[k][k]$ is free to be selected by the optimizer, as noted in Remark~\ref{remark:free_time}.} Finally, we recorded an average run time of $\hat{t}_\mathrm{MPFTC}=0.235$s. While the run time might seem large, we ought to stress that our implementation was not based on efficient real-time algorithms which are expected to be much faster.
	
	\section{Conclusions}\label{sec:conclusions}
	In this paper, we have introduced a new predictive control framework solving the problem of tracking reference trajectories, which might become infeasible in the presence of a-priori unknown constraints, while avoiding undesirably aggressive behaviors. We have discussed safety in a general sense, and provided new terminal conditions, which hold for reference trajectories with mild assumptions, and ensure safety and recursive feasibility with respect to a-priori unknown, time-varying constraints. Our proposed approach can be combined with existing motion planning techniques in order to deliver better performance. Future work will focus on practical real-time implementations of the framework in the context of urban autonomous driving, where moving obstacles are detected by the onboard sensors and are a-priori unknown. 
	
	\bibliographystyle{IEEEtran}
	\bibliography{references}

\begin{thebibliography}{10}
\providecommand{\url}[1]{#1}
\csname url@samestyle\endcsname
\providecommand{\newblock}{\relax}
\providecommand{\bibinfo}[2]{#2}
\providecommand{\BIBentrySTDinterwordspacing}{\spaceskip=0pt\relax}
\providecommand{\BIBentryALTinterwordstretchfactor}{4}
\providecommand{\BIBentryALTinterwordspacing}{\spaceskip=\fontdimen2\font plus
\BIBentryALTinterwordstretchfactor\fontdimen3\font minus
  \fontdimen4\font\relax}
\providecommand{\BIBforeignlanguage}[2]{{%
\expandafter\ifx\csname l@#1\endcsname\relax
\typeout{** WARNING: IEEEtran.bst: No hyphenation pattern has been}%
\typeout{** loaded for the language `#1'. Using the pattern for}%
\typeout{** the default language instead.}%
\else
\language=\csname l@#1\endcsname
\fi
#2}}
\providecommand{\BIBdecl}{\relax}
\BIBdecl

\bibitem{Mayne2000}
D.~Mayne, J.~Rawlings, C.~Rao, and P.~Scokaert, ``Constrained model predictive
  control: Stability and optimality,'' \emph{Automatica}, vol.~36, no.~6, pp.
  789 -- 814, 2000.

\bibitem{rawlings2009model}
J.~B. Rawlings and D.~Q. Mayne, \emph{Model predictive control: Theory and
  design}.\hskip 1em plus 0.5em minus 0.4em\relax Nob Hill Pub. Madison,
  Wisconsin, 2009.

\bibitem{borrelli2017predictive}
F.~Borrelli, A.~Bemporad, and M.~Morari, \emph{Predictive control for linear
  and hybrid systems}.\hskip 1em plus 0.5em minus 0.4em\relax Cambridge
  University Press, 2017.

\bibitem{ljungqvist2016path}
O.~Ljungqvist, D.~Axehill, and A.~Helmersson, ``Path following control for a
  reversing general 2-trailer system,'' in \emph{2016 IEEE 55th Conference on
  Decision and Control (CDC)}.\hskip 1em plus 0.5em minus 0.4em\relax IEEE,
  2016, pp. 2455--2461.

\bibitem{ljungqvist2018stability}
O.~Ljungqvist, D.~Axehill, and J.~L{\"o}fberg, ``On stability for state-lattice
  trajectory tracking control,'' in \emph{2018 Annual American Control
  Conference (ACC)}.\hskip 1em plus 0.5em minus 0.4em\relax IEEE, 2018, pp.
  5868--5875.

\bibitem{andersson2018receding}
O.~Andersson, O.~Ljungqvist, M.~Tiger, D.~Axehill, and F.~Heintz,
  ``Receding-horizon lattice-based motion planning with dynamic obstacle
  avoidance,'' in \emph{2018 IEEE Conference on Decision and Control
  (CDC)}.\hskip 1em plus 0.5em minus 0.4em\relax IEEE, 2018, pp. 4467--4474.

\bibitem{ljungqvist2019path}
O.~Ljungqvist, N.~Evestedt, D.~Axehill, M.~Cirillo, and H.~Pettersson, ``A path
  planning and path-following control framework for a general 2-trailer with a
  car-like tractor,'' \emph{Journal of Field Robotics}, vol.~36, no.~8, pp.
  1345--1377, 2019.

\bibitem{paden2016survey}
B.~Paden, M.~{\v{C}}{\'a}p, S.~Z. Yong, D.~Yershov, and E.~Frazzoli, ``A survey
  of motion planning and control techniques for self-driving urban vehicles,''
  \emph{IEEE Transactions on intelligent vehicles}, vol.~1, no.~1, pp. 33--55,
  2016.

\bibitem{mohanan2018survey}
M.~Mohanan and A.~Salgoankar, ``A survey of robotic motion planning in dynamic
  environments,'' \emph{Robotics and Autonomous Systems}, vol. 100, pp.
  171--185, 2018.

\bibitem{fiorini1998motion}
P.~Fiorini and Z.~Shiller, ``Motion planning in dynamic environments using
  velocity obstacles,'' \emph{The International Journal of Robotics Research},
  vol.~17, no.~7, pp. 760--772, 1998.

\bibitem{van2005roadmap}
J.~P. Van Den~Berg and M.~H. Overmars, ``Roadmap-based motion planning in
  dynamic environments,'' \emph{IEEE Transactions on Robotics}, vol.~21, no.~5,
  pp. 885--897, 2005.

\bibitem{fulgenzi2008probabilistic}
C.~Fulgenzi, C.~Tay, A.~Spalanzani, and C.~Laugier, ``Probabilistic navigation
  in dynamic environment using rapidly-exploring random trees and gaussian
  processes,'' in \emph{2008 IEEE/RSJ International Conference on Intelligent
  Robots and Systems}.\hskip 1em plus 0.5em minus 0.4em\relax IEEE, 2008, pp.
  1056--1062.

\bibitem{aguiar2005path}
A.~P. Aguiar, J.~P. Hespanha, and P.~V. Kokotovic, ``Path-following for
  nonminimum phase systems removes performance limitations,'' \emph{IEEE
  Transactions on Automatic Control}, vol.~50, no.~2, pp. 234--239, 2005.

\bibitem{Faulwasser2009}
T.~Faulwasser, B.~Kern, and R.~Findeisen, ``Model predictive path-following for
  constrained nonlinear systems,'' in \emph{Proceedings of the 48th IEEE
  Conference on Decision and Control. CDC, 2009.}\hskip 1em plus 0.5em minus
  0.4em\relax IEEE, 2009, pp. 8642--8647.

\bibitem{kanjanawanishkul2009path}
K.~Kanjanawanishkul and A.~Zell, ``Path following for an omnidirectional mobile
  robot based on model predictive control,'' in \emph{2009 IEEE International
  Conference on Robotics and Automation}.\hskip 1em plus 0.5em minus
  0.4em\relax IEEE, 2009, pp. 3341--3346.

\bibitem{alessandretti2013trajectory}
A.~Alessandretti, A.~P. Aguiar, and C.~N. Jones, ``Trajectory-tracking and
  path-following controllers for constrained underactuated vehicles using model
  predictive control,'' in \emph{2013 european control conference (ecc)}.\hskip
  1em plus 0.5em minus 0.4em\relax IEEE, 2013, pp. 1371--1376.

\bibitem{Faulwasser2016}
T.~{Faulwasser} and R.~{Findeisen}, ``Nonlinear model predictive control for
  constrained output path following,'' \emph{IEEE Transactions on Automatic
  Control}, vol.~61, no.~4, pp. 1026--1039, April 2016.

\bibitem{faulwasser_implementation}
T.~{Faulwasser}, T.~{Weber}, P.~{Zometa}, and R.~{Findeisen}, ``Implementation
  of nonlinear model predictive path-following control for an industrial
  robot,'' \emph{IEEE Transactions on Control Systems Technology}, vol.~25,
  no.~4, pp. 1505--1511, July 2017.

\bibitem{petti2005safe}
S.~Petti and T.~Fraichard, ``Safe motion planning in dynamic environments,'' in
  \emph{2005 IEEE/RSJ International Conference on Intelligent Robots and
  Systems}.\hskip 1em plus 0.5em minus 0.4em\relax IEEE, 2005, pp. 2210--2215.

\bibitem{liu2017provably}
S.~B. Liu, H.~Roehm, C.~Heinzemann, I.~L{\"u}tkebohle, J.~Oehlerking, and
  M.~Althoff, ``Provably safe motion of mobile robots in human environments,''
  in \emph{2017 IEEE/RSJ International Conference on Intelligent Robots and
  Systems (IROS)}.\hskip 1em plus 0.5em minus 0.4em\relax IEEE, 2017, pp.
  1351--1357.

\bibitem{beckert2017online}
D.~Beckert, A.~Pereira, and M.~Althoff, ``Online verification of multiple
  safety criteria for a robot trajectory,'' in \emph{2017 IEEE 56th Annual
  Conference on Decision and Control (CDC)}.\hskip 1em plus 0.5em minus
  0.4em\relax IEEE, 2017, pp. 6454--6461.

\bibitem{Gros2017a}
S.~{Gros} and M.~{Zanon}, ``Penalty functions for handling large deviation of
  quadrature states in nmpc,'' \emph{IEEE Transactions on Automatic Control},
  vol.~62, no.~8, pp. 3848--3860, Aug 2017.

\bibitem{Grune2011}
L.~Gr\"une and J.~Pannek, \emph{{N}onlinear {M}odel {P}redictive
  {C}ontrol}.\hskip 1em plus 0.5em minus 0.4em\relax London: Springer, 2011.

\bibitem{Mayne2014}
D.~Q. Mayne, ``Model predictive control: Recent developments and future
  promise,'' \emph{Automatica}, vol.~50, no.~12, pp. 2967 -- 2986, 2014.

\bibitem{kerrigan2001robust}
E.~C. Kerrigan, ``Robust constraint satisfaction: Invariant sets and predictive
  control,'' Ph.D. dissertation, University of Cambridge, 2001.

\bibitem{yu2013tube}
S.~Yu, C.~Maier, H.~Chen, and F.~Allg{\"o}wer, ``Tube mpc scheme based on
  robust control invariant set with application to lipschitz nonlinear
  systems,'' \emph{Systems \& Control Letters}, vol.~62, no.~2, pp. 194--200,
  2013.

\bibitem{andersson2019casadi}
J.~A. Andersson, J.~Gillis, G.~Horn, J.~B. Rawlings, and M.~Diehl, ``Casadi: a
  software framework for nonlinear optimization and optimal control,''
  \emph{Mathematical Programming Computation}, vol.~11, no.~1, pp. 1--36, 2019.

\bibitem{wachter2006implementation}
A.~W{\"a}chter and L.~T. Biegler, ``On the implementation of an interior-point
  filter line-search algorithm for large-scale nonlinear programming,''
  \emph{Mathematical programming}, vol. 106, no.~1, pp. 25--57, 2006.

\bibitem{Scokaert1999a}
P.~Scokaert and J.~Rawlings, ``{F}easibility {I}ssues in {L}inear {M}odel
  {P}redictive {C}ontrol,'' \emph{AIChE Journal}, vol.~45, no.~8, pp.
  1649--1659, 1999.

\bibitem{MPT3}
M.~Herceg, M.~Kvasnica, C.~Jones, and M.~Morari, ``{Multi-Parametric Toolbox
  3.0},'' in \emph{Proc.~of the European Control Conference}, Z\"urich,
  Switzerland, July 17--19 2013, pp. 502--510,
  \url{http://control.ee.ethz.ch/~mpt}.

\bibitem{boyd2004convex}
S.~Boyd and L.~Vandenberghe, \emph{Convex optimization}.\hskip 1em plus 0.5em
  minus 0.4em\relax Cambridge, UK: Cambridge university press, 2004.

\bibitem{lofberg2004yalmip}
J.~Lofberg, ``Yalmip: A toolbox for modeling and optimization in matlab,'' in
  \emph{2004 IEEE international conference on robotics and automation (IEEE
  Cat. No. 04CH37508)}.\hskip 1em plus 0.5em minus 0.4em\relax IEEE, 2004, pp.
  284--289.

\end{thebibliography}
	\begin{IEEEbiography}[{\includegraphics[width=1in,height=1.25in,clip,keepaspectratio]{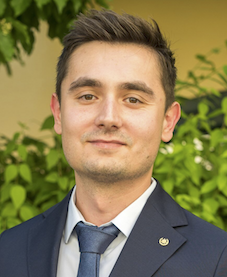}}]{Ivo Batkovic}
		received the Bachelor's and Master's degree in Engineering Physics from Chalmers University of Technology, Gothenburg, Sweden, in 2016. He is currently working towards the Ph.D degree at Chalmers and Zenseact AB, Gothenburg, Sweden. His research interests include model predictive control for constrained autonomous systems in combination with prediction models for safe decision making in autonomous driving applications.
	\end{IEEEbiography}

	\begin{IEEEbiography}[{\includegraphics[width=1in,height=1.25in,clip,keepaspectratio]{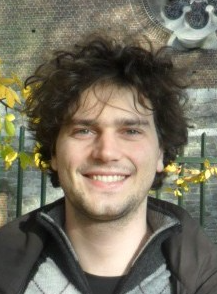}}]{Mario Zanon}
		received the Master's degree in Mechatronics from the University of Trento, and the Dipl\^{o}me d'Ing\'{e}nieur from the Ecole Centrale Paris, in 2010. After research stays at the KU Leuven, University of Bayreuth, Chalmers University, and the University of Freiburg he received the Ph.D. degree in Electrical Engineering from the KU Leuven in November 2015. He held a Post-Doc researcher position at Chalmers University until the end of 2017 and is now Assistant Professor at the IMT School for Advanced Studies Lucca. His research interests include numerical methods for optimization, economic MPC, reinforcement learning, and the optimal control and estimation of nonlinear dynamic systems, in particular for aerospace and automotive applications.
	\end{IEEEbiography}

	\begin{IEEEbiography}[{\includegraphics[width=1in,height=1.25in,clip,keepaspectratio]{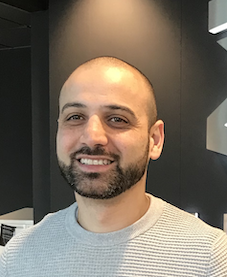}}]{Mohammad Ali}
		received his M.Sc degree in Electrical Engineering and Ph.D. degree in Mechatronics from Chalmers University of Technology in 2005 and 2012, respectively. He is currently responsible for Autonomous driving development at Zenseact AB. His research interests include, methods for safe decision making, trajectory planning with safety guarantees and verification of safety critical perception and control systems.  
	\end{IEEEbiography}

	\begin{IEEEbiography}[{\includegraphics[width=1in,height=1.25in,clip,keepaspectratio]{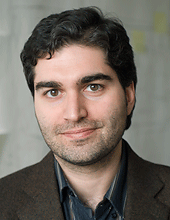}}]{Paolo Falcone}
		received his M.Sc. (“Laurea degree”) in 2003 from the University of Naples “Federico II” and his Ph.D. degree in Information Technology in 2007 from the University of Sannio, in Benevento, Italy. He is Professor at the Department of Electrical Engineering of the Chalmers University of Technology, Sweden. His research focuses on constrained optimal control applied to autonomous and semi-autonomous mobile systems, cooperative driving and intelligent vehicles, in cooperation with the Swedish automotive industry, with a focus on autonomous driving, cooperative driving and vehicle dynamics control.
	\end{IEEEbiography}
	
	\appendix
	
	\subsection{Model details}\label{appendix:model}
	
	The functions used in \eqref{eq:robot} are defined as
	\begin{subequations}\label{eq:modelparams}
		\begin{align}
		B(q) &:= \matr{cc}{b_1+b_2\cos(q_2) & b_3+b_4\cos(q_2)\\ 
			b_3+b_4\cos(q_2) & b_5},\\
		C(q,\dot{q}) &:= -c_1\sin(q_2)\matr{cc}{\dot{q}_1 & \dot{q}_1+\dot{q}_2\\
			-\dot{q}_1	&  0}\\
		g(q) &:= \matr{c}{g_1\cos(q_1)+g_2\cos(q_1+q_2)\\
			g_2\cos(q_1+q_2)},
		\end{align}
	\end{subequations}
	with all parameters given in Table \ref{tab:param}.

	\subsection{Terminal set computation}\label{appendix:terminal}
	In order to construct the terminal region, similarly to~\cite{Faulwasser2016}, we apply the following variable transformation
	\begin{gather}
	\eta_1 = x_1 - p(\theta(\tau)),\quad \eta_2 = x_2 - \frac{\partial{}p}{\partial\theta}\dot{\theta}(\tau),
	\end{gather}
	and rewrite \eqref{eq:robot} as a deviation from the trajectory
	\begin{equation}\label{eq:robot_error}
	\matr{c}{\dot\eta_1\\\dot\eta_2} = \matr{c}{ \eta_2\\ \alpha(\eta,\tau,u)},
	\end{equation}
	\begin{align*}
	\alpha(\eta,\tau,u) = &B^{-1}(\eta,\tau)\Big(u - C(\eta,\tau)(\eta_2+ \frac{\partial{}p}{\partial\theta}\dot{\theta}(\tau))\\
	&-g(\eta_1,\tau)\Big) - \frac{\partial{}^2p}{\partial\theta^2}\dot{\theta}^2(\tau)- \frac{\partial{}p}{\partial\theta}\ddot{\theta}(\tau).
	\end{align*}
	
	Exploiting this form, we define the terminal feedback law
	\begin{align}\label{eq:u_xi}
	\begin{split}
	u_\mathcal{X}(\eta,\tau) = &C(\eta,\tau)\Big(\eta_2+ \frac{\partial{}p}{\partial\theta}\dot{\theta}(\tau)\Big)+g(\eta_1,\tau)\\
	&+B(\eta_1,\tau)(-K_\eta\eta+\ddot{p}(\theta(\tau))),
	\end{split}
	\end{align}
	yielding closed-loop dynamics $\dot{\eta}=(A_\eta-B_\eta{}K_\eta)\eta$, with corresponding Lyapunov function
	\begin{equation}
	\label{eq:lyap_fcn}
	V(\eta) = \eta^\top{}P\eta,\quad P > 0.
	\end{equation}
	We define the terminal region $\mathcal{X}_\eta\in\mathbb{R}^4$ as level set of~\eqref{eq:lyap_fcn}, and use the following bounds
	\begin{gather*}
	\forall x \in \mathcal{X}:\,\left \|B(x_1)\right \|_2\leq{}\bar{B},\, \left \|C(x_1,x_2)\right \|_2\leq{}\bar{C},\,\|g(x_1)\|_2\leq{}\bar{g}.
	\end{gather*}
	The upper bounds  on $\|\dot{p}(\theta(t))\|_2\leq{}\bar{\dot{p}}$ and $\|\ddot{p}(\theta(t))\|_2\leq\bar{\ddot{p}}$ are given directly through the design of the timing law.
	\begin{table} \caption{Model parameters for robotic joint}
		\centering
		\begin{tabular}{|lll||lll|}
			\hline
			$b_1$ &  $200.0$ & [kg m$^2$ / rad] & $b_2$ & $50.0$ 	& [kg m$^2$ / rad]\\
			$b_3$ &  $23.5$ & [kg m$^2$ / rad] & $b_4$ & $25.0$ 	& [kg m$^2$ / rad]\\
			$b_5$ &  $122.5$ & [kg m$^2$ / rad] & $c_1$ & $-25.0$ 	& [Nms$^{-2}$]\\
			$g_1$ & $784.8$ & [Nm] & $g_2$ & $245.3$ & [Nm]\\
			\hline
		\end{tabular}
		\label{tab:param}
	\end{table}
	Furthermore, we tighten the state and input constraints
	\begin{align*}
	\|u\|_2 \leq{} \bar{u},&& \|x_2\|_2\leq{}\bar{\dot{q}}.
	\end{align*}
	In order to obtain a tightened bound on the terminal control input we impose
	\begin{equation}
	\label{eq:robot_terminal_u_bound}
	\|u_\mathcal{X}(\eta,\tau)\|_2 \leq{} \bar{C}\bar{\dot{q}}+\bar{g}+\bar{B}(\bar{\ddot{p}}+\|K_\eta\eta\|_2)\leq{}\bar{u},
	\end{equation}
	which in turn yields
	\begin{equation}
	\label{eq:robot_terminal_u_bound_2}
	\|\eta\|_2 \leq{} \frac{\bar{u}-\bar{C}\bar{\dot{q}}-\bar{g}-\bar{B}\bar{\ddot{p}}}{\bar{B}\|K_\eta\|_2}.
	\end{equation}
	Finally, the terminal set is given as the level set of the Lyapunov function $V(\eta)$
	\begin{equation}
	\mathcal{E}_\eta= \{\eta\in\mathbb{R}^4\, |\, \eta^\top{}P_\eta \eta\leq{}\gamma_*\},
	\end{equation}
	where parameter $\gamma_*$ maximizes the volume of the ellipsoid. Given $P_\eta$ and $K_\eta$, we use the S-procedure \cite{boyd2004convex} to formulate the maximization as convex optimization problem
	\begin{subequations}\label{eq:gamma}
		\begin{align}
		\gamma_*:=\max_{\gamma,\lambda}\quad & \gamma \\
		\mathrm{s.t.} \quad &\matr{cc}{P_\eta & \\ & -\gamma} - \lambda_1 \matr{cc}{I & \\ & -d_1 } \succeq 0,\label{eq:gamma_1} \\
		&\matr{cc}{P_\eta & \\ & -\gamma} - \lambda_2 \matr{cc}{\tilde{I}^0 & \\ & -d_2 } \succeq 0, \label{eq:gamma_2}\\
		& \lambda \geq 0,
		\end{align}
	\end{subequations}			
	where $\tilde{I}^0=\mathrm{blockdiag}(0,0,I)\in\mathbb{R}^4$ and
	\begin{align*}
	d_1=\frac{\bar{u}-\bar{C}\bar{\dot{q}}-\bar{g}-\bar{B}\bar{\ddot{p}}}{\bar{B}\|K_\xi\|_2}, && d_2 = \bar{\dot{q}}-\bar{\dot{p}}.
	\end{align*}
	Constraint \eqref{eq:gamma_1} ensures that~\eqref{eq:robot_terminal_u_bound_2} holds, i.e., the terminal control $u_\mathcal{X}(\xi,\eta)$ satisfies~\eqref{eq:robot_terminal_u_bound}; while constraint \eqref{eq:gamma_2} ensures that  $\|x_2\|_2\leq{}\bar{\dot{q}}$.
	To solve \eqref{eq:gamma}, we use the model data from Table \ref{tab:param} to get the bounds: $\bar{B}=266.4$, $\bar{C}=269.6$ and $\bar{g}=1058.9$. From the timing law we know that $\bar{\dot{p}}=1$ and $\bar{\ddot{p}}=0.823$. We compute the feedback matrix $K_\eta$ via an LQR controller with tuning $Q_\eta=I$, and $R_\eta=10I$ to get 
	\begin{equation}
	K_\eta=\matr{cc}{K_1 & K_2},
	\end{equation}
	where $K_1=0.31\cdot{}\mathrm{diag}(1,1)$ and $K_2=0.85\cdot{}\mathrm{diag}(1,1)$. Finally, the terminal cost is obtained by solving the Lyapunov equation
	\begin{equation*}
	P_\eta = (A_\eta-B_\eta{}K_\eta)^\top P_\eta (A_\eta-B_\eta{}K_\eta) + (Q + K_\eta^\top R K_\eta),
	\end{equation*}
	which gives us
	\begin{equation}
	P_\eta = 10^6\cdot{}\matr{cc}{P_1 &  P_2\\ P_2 & P_3},
	\end{equation}
	with $P_1=6.51\cdot{}\mathrm{diag}(1,1)$, $P_2=5.27\cdot{}\mathrm{diag}(1,1)$, $P_3=6.16\cdot{}\mathrm{diag}(1,1)$. Then, solving~\eqref{eq:gamma} using~\cite{lofberg2004yalmip}, we obtain
	\begin{equation}
	\mathcal{E}_\eta = \{\eta \in\mathbb{R}^4\, |\, \eta^\top P_\eta \eta\leq{}2.29\cdot{}10^7\}.
	\end{equation}
	
	\ifCLASSOPTIONcaptionsoff
	\newpage
	\fi

\end{document}